\documentclass[11pt]{article}
\usepackage[utf8]{inputenc}
\usepackage[margin=1in]{geometry}
\usepackage[affil-it]{authblk}

\usepackage[utf8]{inputenc} 
\usepackage[T1]{fontenc}    
\usepackage{hyperref}       
\usepackage{url}            
\usepackage{booktabs}       
\usepackage{amsfonts}       
\usepackage{stmaryrd}
\usepackage{pifont}
\usepackage{nicefrac}       
\usepackage{microtype}      
\usepackage{lipsum}
\usepackage[numbers]{natbib}
\usepackage[]{url} 
\usepackage{hhline}
\usepackage{appendix}

\usepackage{tikz-cd}
\usepackage{subcaption}
\usepackage{mathtools}
\usepackage{stackengine}
\setstackEOL{\\}
\usepackage{tikz}
\usepackage{tikz-3dplot}
\usepackage{caption}
\usepackage{booktabs}

\usepackage{amssymb}
\usepackage{amsmath,amsfonts,amsthm} 
\usepackage{stmaryrd}
\usepackage{enumerate}
\usepackage{array}
\usepackage{scalerel}

\usepackage{verbatim}
\usepackage{tensor}
\usetikzlibrary{backgrounds}
\usepackage{float}
\usepackage{csquotes}

\newtheorem{theorem}{Theorem}[section]
\newtheorem{example}{Example}[section]
\newtheorem{proposition}{Proposition}[section]
\newtheorem{cor}{Corollary}[section]
\newtheorem{lemma}{Lemma}[section]
\theoremstyle{definition}
\newtheorem{definition}{Definition}[section]

\providecommand{\R}[0]{\mathbb{R}}

\providecommand{\norm}[1]{\lVert#1\rVert}

\providecommand{\inner}[2]{\langle#1,#2\rangle}

\title{Classical Limits of Hilbert Bimodules as Symplectic Dual Pairs}

\author{Benjamin H. ~Feintzeig\\\texttt{bfeintze@uw.edu}}
\affil{{Department of Philosophy}\\{University of Washington}}
\author{Jer Steeger\footnote{These authors contributed equally to this work.}\\\texttt{jsteeger@uw.edu}}
\affil{{Department of Philosophy}\\{University of Washington}}
\date{}

\begin{document}

\maketitle

\begin{abstract}
Hilbert bimodules are morphisms between C*-algebraic models of quantum systems, while symplectic dual pairs are morphisms between Poisson geometric models of classical systems.  Both of these morphisms preserve representation-theoretic structures of the relevant types of models.  Previously, it has been shown that one can functorially associate certain symplectic dual pairs to Hilbert bimodules through strict deformation quantization. We show that, in the inverse direction, strict deformation quantization also allows one to functorially take the classical limit of a Hilbert bimodule to reconstruct a symplectic dual pair.
\end{abstract}



\section{Introduction}
\label{sec:intro}

It is well known that many quantum and classical systems can be associated with a Lie group, whose representations---understood differently relative to the classical or quantum context---capture the structure of conserved quantities of the physical system.  The current paper investigates the relationship between the classical and quantum structures captured by the representations of such a group.  We prove that there is a structural correspondence between a quantum system with conserved quantities and a classical system with conserved quantities constructed through the classical limit.  We make this structural correspondence precise in the form of a functor representing the classical limit that maps between suitable categories of models of quantum and classical physics.  In both categories, the morphisms can be understood as relations between models that preserve the group representation structure, where those representations are defined appropriate to the quantum or classical nature of the model.

On the quantum side, the models we consider are formulated in terms of collections of observables that form C*-algebras.  The morphisms between such models are \emph{Hilbert bimodules}, which provide a procedure for inducing Hilbert space representations of one model from another  \citep{FeDo88,Ma68,Ri72a,Ri74,La95}.
On the classical side, the models we consider are formulated in terms of phase spaces that are Poisson manifolds.  The morphisms between such models are \emph{symplectic dual pairs}, which provide a procedure for inducing symplectic representations or realizations of one model from another \citep{MaWe74,We83,Xu91,Xu92,Xu94a}.

The classical and quantum models are related by \emph{strict deformation quantization} \citep{Ri89,Ri93,La98b,La06,BiGa15}.  This implies that one can form a continuous bundle of C*-algebras \citep{KiWa95} over a base space consisting of values of Planck's constant $\hbar$.  In this setting, at values $\hbar>0$, the fiber C*-algebra is the quantum model given by operators on a Hilbert space.  And at the value $\hbar = 0$, the fiber C*-algebra is an algebra of continuous functions on the phase space of the classical model.  In the opposite direction of quantization, the classical limit can be understood as an operation that generates the C*-algebra of quantities of a classical system from such a continuous bundle \citep{StFe21a}.

Beyond this relationship between classical and quantum models, \citet{La95b,La01a,La05,La02a} provides a relationship between classical and quantum morphisms.  This is given by a quantization map that assigns each suitable classical morphism in the form of a symplectic dual pair to a corresponding quantum morphism in the form of a Hilbert bimodule.  The result is that quantization may be understood as a functor.

In the current paper, we establish a relationship in the opposite direction between quantum and classical morphisms according to the classical $\hbar\to 0$ limit.  We provide a limiting map that assigns each suitable quantum morphism in the form of a Hilbert bimodule to a corresponding classical morphism in the form of a symplectic dual pair.  We show that this map is likewise a classical limit functor.  Indeed, we conjecture that the classical limit functor is almost inverse to the quantization functor \citep{Fe24}.

The goal of the current paper is to provide a two-step procedure for constructing a symplectic dual pair as the classical limit of a Hilbert bimodule between C*-algebras of quantum observables.  In the first step, we follow the strategy of \citet{StFe21a}, who understand classical limits of continuous bundles of C*-algebras as commutative C*-algebras given by a quotient construction; we construct Hilbert bimodules between such commutative C*-algebras via an analogous quotient procedure.  In the second step, we extend the well known Gelfand duality \citep[][\S4.4]{KaRi97} to construct a topological space from our Hilbert bimodule between commutative C*-algebras.  To form a symplectic dual pair, we need to endow this topological space with the canonical structure of a smooth manifold and a symplectic form.  In the literature, \citet{Ne03} already provides necessary and sufficient conditions for a commutative algebra to be dual to a smooth manifold.  Since algebraic conditions for the construction of a smooth structure are already known, we will take these conditions as assumptions, and thus our goal is to focus on the definition of a symplectic form on this space.  Throughout, our primary aim is to establish functoriality of the construction by showing that it respects the composition of Hilbert bimodules and symplectic dual pairs via their corresponding tensor products.

We leave a number of open questions for future research.  For example, we rely on some assumptions about the Hilbert bimodule at $\hbar=0$ in order to guarantee that the dual space carries the structure of a smooth manifold.  It would be interesting to search for further useful conditions on the Hilbert bimodules at $\hbar>0$ that would guarantee  our construction yields a smooth manifold.  Likewise, we are so far able to establish that the construction applies only to a limited class of Hilbert bimodules.  It would be interesting to show that the construction applies to other types of Hilbert bimodules, which might require some generalization beyond the results of this paper.  The current paper thus provides only the first steps in the analysis of classical limits of Hilbert bimodules.

The paper is organized as follows.  In \S\ref{sec:back}, we provide background on strict quantization, Hilbert bimodules, and symplectic dual pairs.  The remainder of the paper provides the two-step procedure to construct a symplectic dual pair as the classical limit of a Hilbert bimodule satisfying certain conditions.  First, in \S\ref{sec:classlim} we prove that one can perform a quotient construction on a Hilbert bimodule between C*-algebras of quantum observables to obtain a Hilbert bimodule between commutative C*-algebras of classical observables.  Second, in \S\ref{sec:dual}  we construct a symplectic manifold that is, in a certain sense, dual to this Hilbert bimodule between commutative C*-algebras when the appropriate conditions are satisfied.  Finally, \S\ref{sec:con} concludes with a discussion and future directions.  

\section{Background}
\label{sec:back}

\subsection{Strict Deformation Quantization}

We will use the tools of strict deformation quantization to relate mathematical models of classical and quantum systems.  We begin by reviewing definitions of the structures involved \citep{Ri89,Ri93,La98b,La06,BiGa15,Fe23}. Throughout, when $X$ is a topological space, we let $C_0(X)$ denote the collection of continuous functions vanishing at infinity, and we let $C_b(X)$ denote the collection of continuous and bounded functions.  Further, if $X$ is a metric space, then we let $UC_b(X)$ denote the collection of uniformly continuous and bounded functions on $X$.  First, we shall need the notion of a strict deformation quantization \citep[][p. 108, Def. 1.1.1]{La98b}.\medskip

\begin{definition}
\label{def:strquant}
Given a Poisson manifold $M$, a \emph{strict quantization} of a Poisson algebra $\mathcal{P}\subseteq C_b(M)$ consists in
\begin{itemize}
\item a locally compact set $I\subseteq \R$ called the \emph{base space}, containing $0$ as an acccumulation point;
\item a family of C*-algebras $(\mathfrak{A}_\hbar)_{\hbar\in I}$ called the \emph{fibers}, with norms denoted by $\norm{\cdot}_\hbar$; and
\item a family of linear, *-preserving \emph{quantization maps} $(\mathcal{Q}_\hbar: \mathcal{P}\to \mathfrak{A}_\hbar)_{\hbar\in I}$, with\newline $\mathcal{Q}_0: \mathcal{P}\hookrightarrow C_b(M)$ the identity and $\mathcal{Q}_\hbar(\mathcal{P})$ dense in $\mathfrak{A}_\hbar$ for each $\hbar\in I$.
\end{itemize}
Together, these structures are required to satisfy the following conditions for every $f,g\in\mathcal{P}$
\begin{enumerate}
\item (von Neumann's condition) $\lim_{\hbar\to 0}\norm{\mathcal{Q}_\hbar(f)\mathcal{Q}_\hbar(g) - \mathcal{Q}_\hbar(fg)}_\hbar = 0$;
\item (Dirac's condition) $\lim_{\hbar\to 0}\norm{\frac{i}{\hbar}[\mathcal{Q}_\hbar(f),\mathcal{Q}_\hbar(g)] - \mathcal{Q}_\hbar(\{f,g\})}_\hbar = 0$;
\item (Rieffel's condition) the map $\hbar\mapsto \norm{\mathcal{Q}_\hbar(f)}_\hbar$ is in $C_0(I)$.
\end{enumerate}
A strict quantization is called a \emph{deformation quantization} if, moreover, each quantization map $\mathcal{Q}_\hbar$ is injective and $\mathcal{Q}_\hbar(\mathcal{P})$ is closed under the product.
\end{definition}\medskip

In the definition of strict quantization, $\mathcal{P}$ may be any Poisson subalgebra of $C_b(M)$.  Note that we only consider bounded functions in the domain of our quantization maps.\medskip

\begin{example}
    Consider the Poisson manifold $\R^{2n}$ with the standard Poisson bracket and the Poisson algebra $\mathcal{P} = C_c^\infty(\R^{2n})$.  Consider the C*-algebra $\mathfrak{A}_0 = C_0(\R^{2n})$ of continuous functions vanishing at infinity and the C*-algebra $\mathfrak{A}_\hbar = \mathcal{K}(L^2(\R^n))$ of compact operators for $\hbar>0$.  Define the maps $\mathcal{Q}_\hbar: \mathcal{P}\to \mathfrak{A}_\hbar$ for $\hbar>0$ by
    \begin{align}
        (\mathcal{Q}_\hbar(f)\psi)(x) = \int_{\R^{2n}} \frac{d^np\ d^nq}{(2\pi\hbar)^n} e^{ip(x-q)/\hbar}f(p,\frac{1}{2}(x+q))\psi(q)
    \end{align}
    for all $\psi\in L^2(\R^n)$ and $f\in \mathcal{P}$.  These maps $(\mathcal{Q}_\hbar)_{\hbar\in [0,1]}$ define a strict deformation quantization known as Weyl quantization \citep[][p. 377, Ex. 23]{Fe23}.\medskip
\end{example}

The Weyl quantization maps can be extended to a wider class of Poisson manifolds associated with Lie groups and Lie groupoids (see \citep{Ri89,La99,BiGa15}).  We recall two examples.  The first concerns the quantization of a symmetry group associated with a physical system.\medskip

\begin{example}
\label{ex:quantLie}
    Let $G$ be a connected compact Lie group and consider the Poisson manifold $\mathfrak{g}^*$ given by the dual of the Lie algebra $\mathfrak{g}$ of $G$.  The Poisson algebra $\mathcal{P} = C_{PW}^\infty(\mathfrak{g}^*)$ is the collection of Paley-Weiner functions, i.e., whose Fourier transform is smooth and compactly supported.  The C*-algebras are $\mathfrak{A}_0 = C_0(\mathfrak{g}^*)$ and $\mathfrak{A}_\hbar = C^*(G)$ (the group C*-algebra of $G$) for $\hbar>0$. \citet[][p. 184, Thm. 1]{La98a} defines quantization maps $\mathcal{Q}_\hbar: \mathcal{P}\to\mathfrak{A}_\hbar$ for $\hbar>0$ and proves that they define a strict deformation quantization.\medskip
\end{example}

Our next example corresponds to a physical system given by a particle moving in an external Yang-Mills field.  The classical phase space of such a system is discussed by \citet{St77} and \citet{We78}.\medskip

\begin{example}
\label{ex:quantYM}
    Let $P\to Q$ be a principal bundle, where the total space $P$ is a Riemannian manifold, with typical fiber given by a connected compact Lie group $G$. Consider the Poisson manifold $(T^*P)/G$.  The Poisson algebra $\mathcal{P} = C_{PW}^\infty((T^*P)/G)$ is the collection of Paley-Weiner functions, i.e., whose fiberwise Fourier transform is smooth and compactly supported.  Consider the C*-algebra $\mathfrak{A}_0 = C_0((T^*P)/G)$ and  the C*-algebra $\mathfrak{A}_\hbar = \mathcal{K}(L^2(Q))\otimes C^*(G)$, for $\hbar>0$.  \citet[][p. 105, Thm. 1]{La93a} defines quantization maps $\mathcal{Q}_\hbar: \mathcal{P}\to\mathfrak{A}_\hbar$ for $\hbar>0$ and proves that they define a strict deformation quantization.\medskip
\end{example}

Recent work has also extended the Weyl quantization maps to different C*-algebras, including the almost periodic functions on the dual of a symplectic vector space \citep{BiHoRi04b,HoRi05,HoRiSc08} and the resolvent algebras \citep{vN19}.  Moreover, the Weyl quantization maps are not unique; on some spaces, one can define a distinct strict deformation quantization known as Berezin quantization \citep[][p. 137, Thm. 2.4.1]{La98b} (see also \citep{BeCo86,Co92}).  Since there can be distinct quantization maps on the same space, one can look for common structure underlying strict quantizations.

Two strict quantizations $(\mathcal{Q}_\hbar)_{\hbar\in I}$ and $(\mathcal{Q}'_\hbar)_{\hbar\in I}$ of a Poisson algebra $\mathcal{P}$ with the same fiber C*-algebras are called \emph{equivalent} if for every $f\in \mathcal{P}$, the map
\begin{align}
\hbar\in I \mapsto \norm{\mathcal{Q}_\hbar(f)-\mathcal{Q}'_\hbar(f)}_\hbar
\end{align}
is continuous \citep[][p. 109]{La98b}.  For example, it is known that Weyl quantization and Berezin quantization on $\R^{2n}$ are equivalent \citep[][p. 144, Prop. 2.6.3]{La98b}.

Equivalent quantizations share a common underlying structure called a bundle of C*-algebras \citep[][p. 677, Def. 1.1]{KiWa95}. 
 While it is more common in quantization theory to employ bundles of C*-algebras satisfying relatively weak continuity constraints \citep[][p. 110, Def. 1.2.1]{La98b}, we will use stronger constraints by requiring uniform continuity in the bundles considered here \citep[][p. 6, Def. 3.1]{StFe21a}.\medskip

\begin{definition}
A \emph{uniformly continuous bundle of C*-algebras} consists in
\begin{itemize}
\item a locally compact metric space $I$ called the \emph{base space};
\item a family of C*-algebras $(\mathfrak{A}_\hbar)_{\hbar\in I}$ called the \emph{fibers}, with norms denoted by $\norm{\cdot}_\hbar$;
\item a C*-algebra $\mathfrak{A}$ called the algebra of \emph{uniformly continuous sections}; and
\item a family of surjective *-homomorphisms $(\phi_\hbar: \mathfrak{A}\to\mathfrak{A}_\hbar)$ called the \emph{evaluation maps}.
\end{itemize}
Together, these structures are required to satisfy the following conditions for every $a\in\mathfrak{A}$:
\begin{enumerate}
\item $\norm{a} = \sup_{\hbar\in I}\norm{\phi_\hbar(a)}_{\hbar}$;
\item for each $f\in UC_b(I)$, there is a continuous section $fa\in \mathfrak{A}$ satisfying $\phi_\hbar(fa) = f(\hbar)\phi_\hbar(a)$ for each $\hbar\in I$;
\item the map $\hbar\mapsto \norm{\phi_\hbar(a)}_\hbar$ is in $UC_b(I)$.
\end{enumerate}
\end{definition}

The usual definition of a continuous bundle of C*-algebras is as above, but replaces each of the appearances of $UC_b(I)$ with $C_0(I)$.  Hence, the only difference in our definition of a uniformly continuous bundle of C*-algebras is that it employs a metric on the base space and requires uniform continuity of the sections.  For more on the comparison of different continuity constraints for bundles of C*-algebras, see \citet[][p. 26, Appendix B]{StFe21a}.

A strict quantization determines a continuous bundle of C*-algebras as long as it satisfies mild technical conditions.  The following construction produces a continuous bundle of C*-algebras just in case for every polynomial $P$ of the maps $[\hbar\mapsto\mathcal{Q}_\hbar(f)]$ for $f\in\mathcal{P}$, the map
\begin{align}
    \hbar\mapsto \norm{P}_\hbar
\end{align}
is continuous \citep[][p. 332]{BiHoRi04b}.  One generates a *-algebra $\tilde{\mathfrak{A}}\subseteq \prod_{\hbar\in I}\mathfrak{A}_\hbar$ of sections by taking pointwise products and sums of maps of the form
\begin{align}
[\hbar\mapsto \mathcal{Q}_\hbar(f)]
\end{align}
for $f\in \mathcal{P}$.  One defines a C*-algebra of continuous sections by
\begin{align}
\mathfrak{A} = \Big\{a\in \prod_{\hbar\in I}\mathfrak{A}_\hbar\ | \ \text{ the map $[\hbar\mapsto \norm{a(\hbar) - \tilde{a}(\hbar)}_\hbar]$ is in $C_0(I)$ for each $\tilde{a}\in\tilde{\mathfrak{A}}$}\Big\}.
\end{align}
With the evaluation maps given by
\begin{align}
\phi_\hbar(a) = a(\hbar)
\end{align}
for each $a\in\mathfrak{A}$ and $\hbar\in I$, this structure becomes a continuous bundle of C*-algebras.  Indeed, 
 one can show that equivalent quantizations determine the same algebra of continuous sections, and hence the same continuous bundle of C*-algebras \citep[][p. 111, Thm. 1.2.4]{La98b}.

We are primarily interested in the case where one has a strict quantization over the base space $[0,1]$, which determines a continuous bundle of C*-algebras over $[0,1]$.  Then the restriction of a continuous bundle over the closed interval $[0,1]$ to the half open interval $(0,1]$ is a uniformly continuous bundle of C*-algebras.  We understand a uniformly continuous bundle of C*-algebras over $(0,1]$ to encode information about the algebras of quantum observables and their scaling properties for values $\hbar>0$.  This will be our starting point in the remainder of the paper.

We consider the classical $\hbar\to 0$ limit as the ``inverse" process to quantization.  As a prototype for the $\hbar\to 0$ limit, we consider the extension of a uniformly continuous bundle of C*-algebras over a base space given by the half-open interval $(0,1]$ to the closed interval $[0,1]$, which includes the value $\hbar=0$.  Since $[0,1]$ is the one-point compactification of the half-open interval $(0,1]$, we treat the $\hbar\to 0$ limit as a particular instance of the extension of a uniformly continuous bundle of C*-algebras over a locally compact base space to the one-point compactification of that base space.

More precisely, suppose one is given a uniformly continuous bundle of C*-algebras over a locally compact metric space $I$ with the algebra of uniformly continuous sections $\mathfrak{A}$, fibers $(\mathfrak{A}_\hbar)_{\hbar\in I}$ and evaluation maps $(\phi_\hbar)_{\hbar\in I}$.  Now, consider the enlarged base space given by the one-point compactification $\dot{I} = I \cup \{0_I\}$ where we denote the point at infinity by $0_I$.  We likewise assume that $\dot{I}$ is a metric space.  \citet[][p. 9, Thm. 4.1 and p. 11, Thm. 4.2]{StFe21a} show that if the inclusion $I\hookrightarrow\dot{I}$ is isometric, then this bundle can be uniquely extended to the one-point compactification $\dot{I}$.  The fiber algebra $\mathfrak{A}_\hbar$ at any $\hbar \in \dot{I}$ (including the limit point $0_I$) is determined by
\begin{align}
\mathfrak{A}_{\hbar} = \mathfrak{A}/K_{\hbar}
\end{align}
where one takes the quotient by the closed two-sided ideal
\begin{align}
K_\hbar = \{ a\in\mathfrak{A}\ |\ \lim_{\hbar'\to \hbar}\norm{\phi_{\hbar'}(a)}_{\hbar'} = 0\}.
\end{align}
In particular, one has the unique limiting fiber algebra $\mathfrak{A}_{0_I} = \mathfrak{A}/K_{0_I}$.

When the uniformly continuous bundle is determined by a strict deformation quantization of a Poisson algebra, with the base space $I=(0,1]$, we have that $\mathfrak{A}_1$ represents the fully quantum theory and $0_I = 0$ so that $\mathfrak{A}_0$ represents the classical theory.  Indeed, in this case where the uniformly continuous bundle $\mathfrak{A}$ is generated by a strict deformation quantization of the Poisson algebra $\mathcal{P} = C_c^\infty(M)$ on a Poisson manifold $M$, the following statements are guaranteed to hold concerning the classical limit:
\begin{itemize}
    \item \emph{The classical limit algebra $\mathfrak{A}_0$ is abelian.}

    The uniqueness of the classical limit \citep[][p. 11, Thm. 4.2]{StFe21a} implies that $\mathfrak{A}_0$ is *-isomorphic to the abelian algebra $C_0(M)$ and hence, the classical limit algebra $\mathfrak{A}_0$ is abelian.  In this case, we have the identification of $M$ with the pure state space $\mathcal{P}(\mathfrak{A}_0)$ as topological spaces by Gelfand duality.  Another way to see this is to note that von Neumann's condition in Def. \ref{def:strquant} implies that in any uniformly continuous bundle generated by a strict quantization, the commutators of all uniformly continuous sections vanish in norm in the limit $\hbar\to 0$, which again implies that the classical limit algebra $\mathfrak{A}_0$ is abelian, and hence is an algebra of continuous functions on a topological space.\smallskip

    \item \emph{The classical limit space $M\cong \mathcal{P}(\mathfrak{A}_0)$ carries the structure of a smooth manifold.}

    It follows from \citet[][p. 22, Prop. 5.5]{StFe21a} that $C_c^\infty(M) \cong \mathcal{Q}(\mathcal{P})/K_0$ (Here $\mathcal{Q}$ is the global quantization map, which takes each function $f\in\mathcal{P}$ to its corresponding continuous section in $\mathfrak{A}$ determined by $\phi_\hbar(\mathcal{Q}(f)) = \mathcal{Q}_\hbar(f)$).  The algebra $C_c^\infty(M)$ is known to determine the smooth structure on $M$ by the construction of \citet{Ne03}, which we describe briefly here.  Following \citet[][p. 34-35]{Ne03}, we define the \emph{smooth envelope} of an algebra $\mathcal{P}$ of functions on a space $X$ as the algebra of all functions on the same space of the form
\begin{align}
g(f_1,...,f_n)(x) = g(f_1(x),...,f_n(x))
\end{align}
for all $x\in X$, where $f_1,...,f_n\in\mathcal{P}$ and $g\in C^\infty(\R^n)$.  We denote the smooth envelope of $\mathcal{P}$ by $\overline{\mathcal{P}}$.  Then for the smooth manifold $M$, we have $\overline{C_c^\infty(M)}\cong C^\infty(M)$.  Likewise, we have that the isomorphic algebra determined by the classical limit satisfies $\overline{\mathcal{Q}(\mathcal{P})/K_0} \cong C^\infty(M)$.  The latter algebra $C^\infty(M)$ (and hence, all algebras isomorphic to it) satisfies the conditions provided by Nestruev of being smooth \citep[][p. 37]{Ne03}, geometric \citep[][p. 23]{Ne03}, and complete \citep[][p. 31]{Ne03}.  Nestruev shows that these are necessary and sufficient conditions to guarantee that $M\cong \mathcal{P}(\mathfrak{A}_0)$ is a smooth manifold  \citep[][p. 77, Thm. 7.2 and p. 79, Thm. 7.7]{Ne03}.  To save space, we omit further discussion of Nestruev's conditions and direct the reader to the reference for a detailed exposition.
    \smallskip

    \item \emph{The classical limit space $M\cong \mathcal{P}(\mathfrak{A}_0)$ carries the structure of a Poisson bracket.}

    The $\hbar\to 0$ classical limit of the commutator in the fiber C*-algebras $(\mathfrak{A}_\hbar)_{\hbar\in I}$ determines a Poisson bracket on $M$ according to the definition given in \citet[][p. 20, Eq. (5.9)]{StFe21a}.  Hence, the structure of $M$ as a Poisson manifold is determined by the classical $\hbar\to 0$ limit.
\end{itemize}

\noindent Thus, the classical limit of a uniformly continuous bundle generated by a strict deformation quantization allows one to recover the geometric structures of the original space one quantized, as determined by structures at $\hbar>0$.

\subsection{Morphisms}

Strict quantization relates the Poisson manifolds used in classical physics to the C*-algebras used in quantum physics.  We will understand these Poisson manifolds and C*-algebras as objects in suitable categories.  Strict quantization can be understood to provide a map between objects that forms a part of a quantization functor \citep[][p. 18, Thm. 1]{La02a}.  We now define the kinds of morphisms we will consider in our categories.

On the classical side, morphisms are given by \emph{symplectic dual pairs}.  (We primarily follow \citep[][p. 542]{We83} and \citep[][p. 168, Def. 5.1]{La01b}, but see also \citep{Xu91,Xu92,Xu94a,La01b,La02a}).\medskip

\begin{definition}
A \emph{symplectic dual pair} $M\leftarrow S\to N$ between Poisson manifolds $M$ and $N$ consists in a symplectic manifold $S$ and smooth Poisson maps $J_M: S\to M$ and $J_N: S\to N$, where $N$ is given minus the Poisson bracket.  Together, they must satisfy for every $f\in C^\infty(M)$ and $g\in C^\infty(N)$
\begin{align}
\{J_M^*f,J_N^*g\} = 0,
\end{align}
where the Poisson bracket is determined by the symplectic form on $S$.\medskip
\end{definition}

In the definition of a symplectic dual pair, no further conditions on the Poisson or symplectic manifolds are required, although in order to guarantee the existence of symplectic dual pairs that can serve as identity morphisms from an object to itself, one restricts to integrable Poisson manifolds \citep[][p. 172, Lemma 5.12]{La01b} (See also \citep{CrFe04}).  We will sometimes focus on particular dual pairs in order to define a notion of composition.  We call a symplectic dual pair \emph{weakly regular} if the maps $J_M$ and $J_N$ are complete, and $J_N$ is a surjective submersion \citep[][cf. p. 171, Def. 5.10]{La01b}.  Here, we say $J_M$ is complete if whenever $f\in C^\infty(M)$ has complete Hamiltonian flow, then $J_M^*f$ has complete Hamiltonian flow on $S$, and likewise for $J_N$ \citep[][p. 67]{La98b}.  We call another dual pair $M\leftarrow \tilde{S}\to N$ with maps $\tilde{J}_M: \tilde{S}\to M$ and $\tilde{J}_N: \tilde{S}\to N$ isomorphic to the dual pair $M\leftarrow S\to N$ denoted above when there is a symplectomorphism $u: S\to\tilde{S}$ such that $\tilde{J}_M \circ u = J_M$ and $\tilde{J}_N\circ u = J_N$.  Sometimes we will refer to a symplectic dual pair by its middle space $S$, leaving the relevant mappings implied.  

Symplectic dual pairs are used to analyze symplectic realizations of Poisson manifolds.  A symplectic realization of a Poisson manifold $P$ is a symplectic manifold $S$ with a smooth Poisson map $J: S\to P$.  It is well known that a symplectic dual pair allows one to induce symplectic realizations of $N$ from those of $M$ \citep[][p. 292, Thm. 2]{La95b}.\medskip

\begin{example}
\label{ex:pmorphdualpair}
Suppose $J: N\to M$ is a smooth, injective Poisson immersion from a connected, simply connected symplectic manifold $N$ to a Poisson  manifold $M$.  In this case, $J$ defines a symplectomorphism between $N$ and $J(N)$.  Then for $S = J(N)\times N$ (where the second space $N$ is given minus the Poisson bracket) the map $J$ defines a symplectic dual pair
\begin{align}
    M\leftarrow J(N)\times N \rightarrow N
\end{align}
where $J_M = pr_M$ and $J_N = pr_N$ are the projection maps.  Thus, every smooth, injective Poisson immersion between connected, simply connected symplectic manifolds corresponds to a symplectic dual pair.  (See also \citet[][p. 50, Remark 4.1]{BuWe04} or \citet[][p. 285, Remark 8.9]{La01c} for a general construction of a dual pair for any smooth Poisson map.)\medskip
\end{example}

\begin{example}
\label{ex:classYMdualpair}
    Let $P\to Q$ be a principal bundle, where $P$ is a Riemannian manifold, with typical fiber given by a connected compact Lie group $G$.  Consider the Poisson manifolds $(T^*P)/G$ and $\mathfrak{g}^*$. 
 \citet[][p. 312]{La95b} constructs a symplectic dual pair 
 \begin{align}
     (T^*P)/G\leftarrow T^*P\rightarrow \mathfrak{g}^*.
 \end{align}
 This dual pair establishes a bijective correspondence between symplectic leaves in $(T^*P)/G$ and symplectic leaves in $\mathfrak{g}^*$, the latter of which are given by coadjoint orbits.\medskip
\end{example}

Symplectic dual pairs that are weakly regular can be composed by a kind of tensor product as follows \citep[][p. 159, Lemma 5.3]{La01b}.  Consider two symplectic dual pairs\newline $M\leftarrow S_1\to N$ with maps $\overset{1}{J}_M: S_1\to M$ and $\overset{1}{J}_N: S_1\to N$, and $N\leftarrow S_2\to P$ with maps $\overset{2}{J}_N: S_2\to N$ and $\overset{2}{J}_P: S_2\to P$.  Let
\begin{align}
    S_1\times_N S_2 = \Big\{(s_1,s_2)\in S_1\times S_2\ |\ \overset{1}{J}_N(s_1) = \overset{2}{J}_N(s_2)\Big\}.
    \end{align}
Since $S_1\times S_2$ carries the product symplectic structure determined by $S_1$ and $S_2$, we can consider the null distribution $\mathcal{N}_{S_1\times_NS_2}$ of vectors in $T(S_1\times_NS_2)$ that are null relative to the symplectic form.  Then define $S_1\circledcirc_N S_2 = (S_1\times_NS_2)/\mathcal{N}_{S_1\times_NS_2}$.  This is a symplectic manifold that defines a symplectic dual pair $M\leftarrow S_1\circledcirc_N S_2\to P$, which we consider the composition of the previous two morphisms \citep[][p. 6]{La02a}.

On the quantum side, morphisms are given by \emph{Hilbert bimodules} (We primarily follow \citep[][p. 144, Def. 3.1]{La01b}, but see also \citep{Ri74,Ri72a,La95,RaWi98,La01a}).\medskip

\begin{definition}
A Hilbert bimodule $\mathfrak{A}\rightarrowtail\mathcal{E}\leftarrowtail\mathfrak{B}$ between C*-algebras $\mathfrak{A}$ and $\mathfrak{B}$ consists in a vector space $\mathcal{E}$ carrying a $\mathfrak{B}$-valued positive definite inner product $\inner{\cdot}{\cdot}_\mathfrak{B}$ that is antilinear in the first argument and linear in the second.  Further, $\mathcal{E}$ has a right action of $\mathfrak{B}$ and a left action of $\mathfrak{A}$ by adjointable operators (Here, a linear operator $a$ on $\mathcal{E}$ is called adjointable \citep[][p. 8]{La95} if there is a linear operator $a^*$ such that $\inner{e_1}{ae_2}_\mathfrak{B} = \inner{a^*e_1}{e_2}_\mathfrak{B}$ for all $e_1,e_2\in\mathcal{E}$.)  The inner product is required to satisfy
\begin{align}
\inner{e_1}{e_2b}_\mathfrak{B} = \inner{e_1}{e_2}_\mathfrak{B}\cdot b
\end{align}
for all $e_1,e_2\in\mathcal{E}$ and $b\in\mathfrak{B}$.  The space $\mathcal{E}$ is required to be complete in the norm 
\begin{align}
    \norm{e}^2 = \norm{\inner{e}{e}_\mathfrak{B}},
\end{align}
where the right hand side uses the C*-norm in $\mathfrak{B}$.
Finally, the $\mathfrak{A}$-action is required to be nondegenerate in the sense that $\mathfrak{A}\mathcal{E}$ is dense in $\mathcal{E}$.
\end{definition}\medskip

We do not require that Hilbert bimodules be \emph{full}, in the sense that we do not require the ideal $\{\inner{e_1}{e_2}_\mathfrak{B}\ |\ e_1,e_2\in\mathcal{E}\}$ be dense in $\mathfrak{B}$.  We call another Hilbert bimodule $\mathfrak{A}\rightarrowtail \tilde{\mathcal{E}}\leftarrowtail\mathfrak{B}$ unitarily equivalent to the Hilbert bimodule $\mathfrak{A}\rightarrowtail\mathcal{E}\leftarrowtail\mathfrak{B}$ denoted above when there is a unitary map $U:\mathcal{E}\to\tilde{\mathcal{E}}$ intertwining both the left $\mathfrak{A}$-actions and the right $\mathfrak{B}$-actions.  Sometimes we will refer to a Hilbert bimodule by its carrying space $\mathcal{E}$, leaving the relevant actions and inner product implied.  It is well known that a Hilbert bimodule allows one to induce Hilbert space representations of $\mathfrak{B}$ from those of $\mathfrak{A}$ \citep[][p. 299]{La95b} (See also \citep{Ri72a,Ri74a,Ri74}).\medskip

\begin{example}
\label{ex:*hombimod}
Suppose $\beta: \mathfrak{A}\to\mathfrak{B}$ is a *-homomorphism between C*-algebras $\mathfrak{A}$ and $\mathfrak{B}$ that is non-degenerate, i.e., $\beta(\mathfrak{A})\mathfrak{B}$ is dense in $\mathfrak{B}$.  Then with $\mathcal{E} = \mathfrak{B}$, the map $\beta$ defines a Hilbert bimodule
\begin{align}
\mathfrak{A}\rightarrowtail\mathfrak{B}\leftarrowtail\mathfrak{B},
\end{align}
where $\mathfrak{B}$ acts on the right on itself by right multiplication, $\mathfrak{A}$ acts on the left on $\mathfrak{B}$ by $ab = \beta(a)b$ for every $a\in\mathfrak{A}$ and $b\in\mathfrak{B}$, and the $\mathfrak{B}$-valued inner product is given by $\inner{b_1}{b_2}_\mathfrak{B} = b_1^*b_2$ for all $b_1,b_2\in\mathfrak{B}$.  Thus, every non-degenerate *-homomorphism corresponds to a Hilbert bimodule.\medskip 
\end{example}

\begin{example}
\label{ex:quantYMbimod}
Let $P\to Q$ be a principal bundle, where $P$ is a Riemannian manifold, with typical fiber given by a connected compact Lie group $G$.  Consider the C*-algebras  $\mathcal{K}(L^2(Q))\otimes C^*(G)$ and $C^*(G)$ as determined by strict quantization in Ex. \ref{ex:quantLie} and \ref{ex:quantYM}.  \citet[][p. 313]{La95b} constructs a Hilbert bimodule
\begin{align}
    \mathcal{K}(L^2(Q))\otimes C^*(G)\rightarrowtail L^2(P)\leftarrowtail C^*(G).
\end{align}
This bimodule establishes a bijective correspondence between irreducible representations of $\mathcal{K}(L^2(Q))\otimes C^*(G)$ and irreducible representations of $C^*(G)$, the latter of which are given by unitary representations of $G$.\medskip
\end{example}

We establish some notation for dealing with a Hilbert bimodule $\mathcal{E}$.  We denote the set of adjointable operators on $\mathcal{E}$ by $\mathcal{L}(\mathcal{E})$.  Consider also the specific adjointable operators $\Theta_{\varphi,\eta}$ on $\mathcal{E}$ for $\varphi,\eta\in\mathcal{E}$ defined by
\begin{align}
\Theta_{\varphi,\eta}(\xi) = \eta\cdot \inner{\varphi}{\xi}_{\mathfrak{B}}
\end{align}
for all $\xi\in\mathcal{E}$.  The closed linear subspace of $\mathcal{L}(\mathcal{E})$ spanned by the operators of the form $\Theta_{\varphi,\eta}$ is called the \emph{compact operators} and denoted by $\mathcal{K}(\mathcal{E})$.

Hilbert bimodules can also be composed by a tensor product as follows.  Consider two Hilbert bimodules $\mathfrak{A}\rightarrowtail \mathcal{E}\leftarrowtail \mathfrak{B}$ and $\mathfrak{B}\rightarrowtail\mathcal{F}\leftarrowtail\mathfrak{C}$.  Let $\mathcal{N}_\mathfrak{B}$ denote the subspace of the algebraic tensor product $\mathcal{E}\otimes\mathcal{F}$ spanned by $\{eb\otimes f - e\otimes bf\ | e\in\mathcal{E}, f\in\mathcal{F}, b\in\mathfrak{B}\}$. 
 Then define $\mathcal{E}\dot{\otimes}_{\mathfrak{B}}\mathcal{F} = (\mathcal{E}\otimes\mathcal{F})/\mathcal{N}_{\mathfrak{B}}$.  This vector space carries a natural left $\mathfrak{A}$-action and right $\mathfrak{C}$-action, as well as a $\mathfrak{C}$-valued inner product determined by
 \begin{align}
\inner{e_1\otimes f_1}{e_2\otimes f_2}_{\mathfrak{C}} = \inner{f_1}{\inner{e_1}{e_2}_{\mathfrak{B}}f_2}_{\mathfrak{C}}
 \end{align}
 for $e_1,e_2\in\mathcal{E}$ and $f_1,f_2\in\mathcal{F}$.  We denote the completion of $\mathcal{E}\dot{\otimes}_{\mathfrak{B}}\mathcal{F}$ by $\mathcal{E}\otimes_\mathfrak{B}\mathcal{F}$, which thus defines a Hilbert bimodule $\mathfrak{A}\rightarrowtail\mathcal{E}\otimes_\mathfrak{B}\mathcal{F}\leftarrowtail\mathfrak{C}$ called the interior tensor product, which we consider the composition of the previous two morphisms \citep[][p. 145]{La01b}.

\citet{La95b,La01a,La05,La02a} shows that one can associate certain symplectic dual pairs between classical Poisson manifolds with Hilbert bimodules between their corresponding quantum C*-algebras.  As a particular case, his construction associates to each dual pair
\begin{align}
    (T^*P)/G\leftarrow T^*P\rightarrow \mathfrak{g}^*
\end{align}
from the construction in Ex. \ref{ex:classYMdualpair} \citep[][p. 110, Lemma 3]{La01a} a corresponding Hilbert bimodule
\begin{align}
    \mathcal{K}(L^2(Q))\otimes C^*(G)\rightarrowtail L^2(P)\leftarrowtail C^*(G)
\end{align}
from the construction in Ex. \ref{ex:quantYMbimod} \citep[][p. 106, Lemma 1]{La01a}.
This quantization of morphisms respects composition of dual pairs \citep[][p. 111, Thm. 3]{La01a} and Hilbert bimodules \citep[][p. 106, Thm. 2]{La01a}, so that quantization may be understood as a functor between suitable categories \citep[][p.18, Thm. 1]{La02a}.  We now turn to an investigation of the correspondence in the opposite direction through the classical limit.  We will establish that the classical limit can be understood as a functor for these same objects and morphisms.

\section{Classical Limit of a Hilbert Bimodule}
\label{sec:classlim}

We now consider the classical limit as the ``inverse" process of quantization, which we will use to define a functor.  We have already seen that the classical limit provides a map from quantum to classical objects.  We now show that, under suitable conditions, one can also take the classical limit of a quantum morphism in the form of a Hilbert bimodule to obtain a classical morphism in the form of a symplectic dual pair.

We proceed in this section by dealing with the general case where we have two continuous bundles of C*-algebras $((\mathfrak{A}_\hbar,\phi_\hbar)_{\hbar\in I},\mathfrak{A})$ and $((\mathfrak{B}_\hbar,\psi_\hbar)_{\hbar\in J},\mathfrak{B})$ over locally compact metric spaces $I$ and $J$.  We consider extending to the one-point compactifications $\dot{I} = I \cup \{0_I\}$ and $\dot{J} = I \cup \{0_J\}$, which we assume are both metric spaces, and we assume that the inclusions $I\hookrightarrow\dot{I}$ and $J\hookrightarrow\dot{J}$ are isometric. 
 Recall that to take the classical limit of objects, we extend these bundles to the base spaces $\dot{I}$ and $\dot{J}$, setting
    \begin{align}
K_{0_I} = \{ a\in\mathfrak{A}\ |\ \lim_{\hbar\to 0_I}\norm{\phi_{\hbar}(a)}_{\hbar} = 0\}\\
K_{0_J} = \{ b\in\mathfrak{B}\ |\ \lim_{\hbar\to 0_J}\norm{\psi_{\hbar}(b)}_{\hbar} = 0\},
\end{align}
and defining $\mathfrak{A}_{0_I} = \mathfrak{A}/K_{0_I}$ and $\mathfrak{B}_{0_J} = \mathfrak{B}/K_{0_J}$.

Suppose now that we have a continuous proper map $\alpha: I\to J$, and for some $\hbar_1\in I$, we have a Hilbert bimodule
\begin{align}
\mathfrak{A}_{\hbar_1}\rightarrowtail\mathcal{E}_{\hbar_1}\leftarrowtail\mathfrak{B}_{\alpha(\hbar_1)}
\end{align}
between the corresponding fiber algebras.  To take the classical limit of the morphism given by $\mathcal{E}_{\hbar_1}$, we will first lift it to obtain a morphism between the algebras of uniformly continuous sections $\mathfrak{A}$ and $\mathfrak{B}$ and then factor through the quotient to obtain a morphism  between the limiting fiber algebras $\mathfrak{A}_{0_I}$ and $\mathfrak{B}_{0_J}$.  With that classical Hilbert bimodule in hand, we will proceed in the next section to construct a symplectic dual pair between the Gelfand dual spaces of those classical algebras.

Here, we first describe the process of lifting the morphism given by $\mathcal{E}_{\hbar_1}$ to a morphism that is a Hilbert bimodule between the algebras of uniformly continuous sections $\mathfrak{A}$ and $\mathfrak{B}$:

\begin{itemize}
\item In the case where the continuous bundles of C*-algebras are generated by strict deformation quantizations, one can take a single Hilbert bimodule $\mathcal{E}_{\hbar_1}$ and use the quantization maps to generate a family of Hilbert bimodules $(\mathcal{E}_\hbar)_{\hbar\in I}$ as follows.  We associate with each strict deformation quantization $(\mathcal{Q}_\hbar)_{\hbar\in I}$ on a domain $\mathcal{P}$ a collection of \emph{rescaling maps} $R^\mathcal{Q}_{\hbar\to\hbar'}: \mathcal{Q}_\hbar(\mathcal{P})\to\mathcal{Q}_{\hbar'}(\mathcal{P})$ defined by
\begin{align}
R^\mathcal{Q}_{\hbar\to\hbar'} = \mathcal{Q}_{\hbar'}\circ (\mathcal{Q}_\hbar)^{-1}
\end{align}
for any $\hbar,\hbar'\in I$.  Suppose one has two strict deformation quantizations $(\mathfrak{A}_\hbar,\mathcal{Q}_\hbar)_{\hbar\in I}$ and $(\mathfrak{B}_\hbar,\mathcal{Q}'_\hbar)_{\hbar\in J}$ generating the uniformly continuous bundles $\mathfrak{A}$ and $\mathfrak{B}$.   For arbitrary $\hbar\in I$, consider the vector space $\tilde{\mathcal{E}}_\hbar = \mathcal{E}_{\hbar_1}$ with the left $\mathfrak{A}_\hbar$-action and right $\mathfrak{B}_{\alpha(\hbar)}$-action on $\tilde{\mathcal{E}}_{\hbar}$ given by
\begin{align}
a_\hbar\cdot \varphi &= R^\mathcal{Q}_{\hbar\to\hbar_1}(a_\hbar)\cdot \varphi\nonumber\\
\varphi\cdot b_{\alpha(\hbar)} &= \varphi\cdot R^{\mathcal{Q}'}_{\alpha(\hbar)\to\alpha(\hbar_1)}(b_{\alpha(\hbar)})
\end{align}
for all $a_\hbar\in\mathfrak{A}_\hbar$, $b_{\alpha(\hbar)}\in\mathfrak{B}_{\alpha(\hbar)}$, and $\varphi\in\tilde{\mathcal{E}}_{\hbar}$. 
If the function $\inner{\cdot}{\cdot}_{\mathfrak{B}_{\alpha(\hbar)}}$ on $\tilde{\mathcal{E}}_\hbar$ given by
\begin{align}
\inner{\varphi}{\eta}_{\mathfrak{B}_{\alpha(\hbar)}} = R^{\mathcal{Q}'}_{\alpha(\hbar_1)\to\alpha(\hbar)}\big(\inner{\varphi}{\eta}_{\mathfrak{B}_{\alpha(\hbar_1)}}\big)
\end{align}
for all $\varphi,\eta\in\tilde{\mathcal{E}}_\hbar$ defines a $\mathfrak{B}_{\alpha(\hbar)}$-valued inner product, and if with $\mathcal{E}_\hbar = \overline{\tilde{\mathcal{E}}_\hbar}$ the norm completion of $\tilde{\mathcal{E}}_\hbar$ one has that the structure defined for each $\hbar\in I$ by
\begin{align}
\label{eq:fibermod}
\mathfrak{A}_\hbar\rightarrowtail\mathcal{E}_\hbar\leftarrowtail\mathfrak{B}_{\alpha(\hbar)}
\end{align}
is a Hilbert bimodule, then we say that the original Hilbert bimodule $\mathcal{E}_{\hbar_1}$ is \emph{scaling}.  

\item Suppose now that one has a continuous field of Banach spaces $(\mathcal{E}_\hbar)_{\hbar\in I}$, where each $\mathcal{E}_\hbar$ is as in Eq. (\ref{eq:fibermod}) and with a space of continuous vector fields given by a subspace $\Gamma\subseteq \prod_{\hbar\in I}\mathcal{E}_\hbar$ (cf. \citep[][p. 211, Def. 10.1.2]{Di77}).  Define the vector space $\mathcal{E}$ as the collection of vector fields $\varphi\in\Gamma$ satisfying
\begin{align}
[\hbar\mapsto \inner{\varphi(\hbar)}{\varphi(\hbar)}_{\mathfrak{B}_{\alpha(\hbar)}}]\in \mathfrak{B}.
\end{align}
and
\begin{align}
[\hbar\mapsto \inner{\phi_\hbar(a)\cdot\varphi(\hbar)}{\phi_\hbar(a)\cdot\varphi(\hbar)}_{\mathfrak{B}_{\alpha(\hbar)}}]\in\mathfrak{B}
\end{align}
for all $a\in\mathfrak{A}$.
In this case, one has a left $\mathfrak{A}$-action and a right $\mathfrak{B}$-action on $\mathcal{E}$ given by
\begin{align}
(a\cdot \varphi)(\hbar) &= \phi_\hbar(a)\cdot \varphi(\hbar)\nonumber\\
(\varphi\cdot b)(\hbar) &= \varphi(\hbar)\cdot \psi_{\alpha(\hbar)}(b)
\end{align}
for all $a\in\mathfrak{A}$, $b\in\mathfrak{B}$, $\varphi\in\mathcal{E}$, and $\hbar\in I$.  If the function $\inner{\cdot}{\cdot}_{\mathfrak{B}}$ on $\mathcal{E}$ given by
\begin{align}
\psi_\hbar\big(\inner{\varphi}{\eta}_{\mathfrak{B}}\big) = \inner{\varphi(\hbar)}{\eta(\hbar)}_{\mathfrak{B}_{\alpha(\hbar)}}
\end{align}
defines a $\mathfrak{B}$-valued inner product for all $\varphi,\eta\in\mathcal{E}$ so that 
\begin{align}
\label{eq:lift}
\mathfrak{A}\rightarrowtail\mathcal{E}\leftarrowtail\mathfrak{B}
\end{align}
forms a Hilbert bimodule, then we say $\mathcal{E}$ is a \emph{lift} of the family $(\mathcal{E}_\hbar)_{\hbar\in I}$ to a morphism between the algebras of uniformly continuous sections $\mathfrak{A}$ and $\mathfrak{B}$.  Note that a lift is taken relative to the continuous field of Banach spaces $\Gamma$.  In concrete examples, an appropriate collection of continuous vector fields $\Gamma$ may be defined based on the context.

\item Finally, we say that the original Hilbert bimodule $\mathcal{E}_{\hbar_1}$ is \emph{continuously scaling} if $\mathcal{E}_{\hbar_1}$ is scaling, so that it defines a family of Hilbert bimodules $(\mathcal{E}_\hbar)_{\hbar\in I}$ as in Eq. (\ref{eq:fibermod}), and if this family of bimodules furthermore has a lift to a Hilbert bimodule $\mathcal{E}$ between $\mathfrak{A}$ and $\mathfrak{B}$ as in Eq. (\ref{eq:lift}).
\end{itemize} 

\begin{example}
\label{ex:bimodfam}
Suppose $\beta: \mathfrak{A}\to\mathfrak{B}$ is a *-homomorphism between the C*-algebras of uniformly continuous sections that is non-degenerate in the sense that $\beta(\mathfrak{A})\mathfrak{B}$ is dense in $\mathfrak{B}$, and also satisfies $\beta(K_\hbar)\subseteq K_{\alpha(\hbar)}$ for each $\hbar\in I$.  It follows \citep[][p. 11, Lemma 4.3]{StFe21} that for each $\hbar\in I$, there is a unique non-degenerate *-homomorphism $\beta_\hbar: \mathfrak{A}_\hbar\to\mathfrak{B}_{\alpha(\hbar)}$ satisfying $\beta_\hbar\circ \phi_\hbar = \psi_{\alpha(\hbar)}\circ \beta$.  At any $\hbar\in I$, consider the Hilbert bimodule 
\begin{align}
\label{eq:1bimod}
\mathfrak{A}_\hbar\rightarrowtail\mathfrak{B}_{\alpha(\hbar)}\leftarrowtail\mathfrak{B}_{\alpha(\hbar)}
\end{align}
associated with $\beta_\hbar$ from Ex. \ref{ex:*hombimod}.  The Hilbert bimodule in Eq. (\ref{eq:1bimod}) is scaling: it associates to each $\hbar'\neq \hbar\in I$ the Hilbert bimodule
\begin{align}
\label{eq:bimodfam}
\mathfrak{A}_{\hbar'}\rightarrowtail\mathfrak{B}_{\alpha(\hbar')}\leftarrowtail\mathfrak{B}_{\alpha(\hbar')}
\end{align}
corresponding to $\beta_{\hbar'}$ through the construction in  Ex. \ref{ex:*hombimod}.  With $\Gamma = \mathfrak{B}$ the space of continuous vector fields, the collection of vector spaces $(\mathfrak{B}_{\alpha(\hbar)})_{\hbar\in I}$ forms a continuous field of Banach spaces.  Relative to this choice of $\Gamma$, the family of Hilbert bimodules in Eq. (\ref{eq:bimodfam}) has a lift to the Hilbert bimodule
\begin{align}
\mathfrak{A}\rightarrowtail\mathfrak{B}\leftarrowtail\mathfrak{B}
\end{align}
associated with the original *-homomorphism $\beta$ through the construction in Ex. \ref{ex:*hombimod}.  Hence, the Hilbert bimodule in Eq. (\ref{eq:1bimod}) is continuously scaling.\medskip
\end{example}

Next, we show that whenever one has a Hilbert bimodule $\mathfrak{A}\rightarrowtail\mathcal{E}\leftarrowtail\mathfrak{B}$ between the algebras of uniformly continuous sections satisfying suitable conditions, it induces a Hilbert bimodule $\mathfrak{A}_{0_I}\rightarrowtail\mathcal{E}_{0_I}\leftarrowtail\mathfrak{B}_{0_J}$ between the fiber algebras at $\hbar = 0_I$ and $\hbar = 0_J$, which can be understood as the classical limit of the Hilbert bimodule.  We subsequently show that this Hilbert bimodule can be used to construct a symplectic dual pair.

\subsection{Quotient of a Hilbert Bimodule}
\label{sec:quotmod}
In this section, we follow the strategy given in \citet[][Prop. 3.25, p. 55]{RaWi98} to construct a Hilbert bimodule at $\hbar = 0_I$ via a quotient.  We again consider the situation where we have two uniformly continuous bundles of C*-algebras $((\mathfrak{A}_\hbar,\phi_\hbar)_{\hbar\in I},\mathfrak{A})$ and $((\mathfrak{B}_\hbar,\psi_\hbar)_{\hbar\in J},\mathfrak{B})$ over locally compact metric spaces $I$ and $J$, whose one-point compactifications $\dot{I}$ and $\dot{J}$ are likewise metric spaces with isometric inclusions $I\hookrightarrow\dot{I}$ and $J\hookrightarrow\dot{J}$.  We will refer to the following ideals in $\mathfrak{A}$ by
\begin{align}
\label{eq:idealA}
K_\hbar &= \{ a\in\mathfrak{A}\ |\ \lim_{\hbar'\to \hbar}\norm{\phi_{\hbar'}(a)}_{\hbar'} = 0\}
\end{align}
for each $\hbar\in \dot{I}$, including the point at infinity $\hbar = 0_I$. Likewise, we will refer to the corresponding ideals in $\mathfrak{B}$ by
\begin{align}
\label{eq:idealB}
K_\hbar &= \{ b\in\mathfrak{B}\ |\ \lim_{\hbar'\to \hbar}\norm{\psi_{\hbar'}(b)}_{\hbar'} = 0\}
\end{align}
for each $\hbar\in \dot{J}$, including the point at infinity $\hbar = 0_J$.  We assume the ambient C*-algebra for each ideal---either $\mathfrak{A}$ or $\mathfrak{B}$---will be clear from context.  We denote the canonical quotient mappings by
\begin{align}
\phi_{0_I}&: \mathfrak{A}\to \mathfrak{A}_{0_I}\nonumber\\
\psi_{0_J}&:\mathfrak{B}\to \mathfrak{B}_{0_J}
\end{align}
with $\mathfrak{A}_{0_I} = \mathfrak{A}/K_{0_I}$ and $\mathfrak{B}_{0_J} = \mathfrak{B}/K_{0_J}$.\medskip

\begin{definition}
Suppose that $\alpha: I\to J$ is a continuous proper map and $\mathfrak{A}\rightarrowtail\mathcal{E}\leftarrowtail\mathfrak{B}$ is a Hilbert bimodule.  We say that $\mathcal{E}$ is \emph{strongly non-degenerate} for $\alpha$ at $\hbar \in \dot{I}$ if we have $\overline{K_{\hbar}\cdot\mathcal{E}} \subseteq \overline{\mathcal{E}\cdot K_{\alpha(\hbar)}}$.  In the special case where $I=J$ and $\alpha$ is the identity map, we say $\mathcal{E}$ is strongly non-degenerate at $\hbar \in \dot{I}$.
\end{definition}\medskip

Given a continuous proper map $\alpha:I\to J$ and a strongly non-degenerate Hilbert bimodule $\mathfrak{A}\rightarrowtail\mathcal{E}\leftarrowtail\mathfrak{B}$ for $\alpha$ at $\hbar\in \dot{I}$, we define:
\begin{itemize}
\item a complex vector space
\begin{align}
\label{eq:limvec}
\mathcal{E}_\hbar = \mathcal{E}/\overline{\mathcal{E}\cdot K_{\alpha(\hbar)}}
\end{align}
with the canonical quotient map
\begin{align}
\gamma: \mathcal{E}\to\mathcal{E}_{\hbar}.
\end{align}
\item a left $\mathfrak{A}_{\hbar}$-action given by
\begin{align}
\label{eq:qAaction}
\phi_{\hbar}(a)\gamma(\varphi) = \gamma(a\cdot\varphi)
\end{align}
for all $a\in\mathfrak{A}$ and $\varphi\in\mathcal{E}$.
\item a right $\mathfrak{B}_{\alpha(\hbar)}$-action given by
\begin{align}
\label{eq:qBaction}
\gamma(\varphi)\cdot\psi_{\alpha(\hbar)}(b) = \gamma(\varphi\cdot b)
\end{align}
for all $b\in\mathfrak{B}$ and $\varphi\in\mathcal{E}$.
\item a bilinear $\mathfrak{B}_{\alpha(\hbar)}$-valued function given by
\begin{align}
\label{eq:qinner}
\inner{\gamma(\varphi)}{\gamma(\eta)}_{\mathfrak{B}_{\alpha(\hbar)}} = \psi_{\alpha(\hbar)}\big(\inner{\varphi}{\eta}_{\mathfrak{B}}\big)
\end{align}
for all $\varphi,\eta\in\mathcal{E}$.
\end{itemize}

\begin{theorem}
\label{thm:quotient}
If $\alpha: I\to J$ is a continuous proper map and $\mathfrak{A}\rightarrowtail\mathcal{E}\leftarrowtail\mathfrak{B}$ is a strongly non-degenerate Hilbert bimodule for $\alpha$ at $\hbar \in \dot{I}$, then the structure
\begin{align}
\mathfrak{A}_{\hbar}\rightarrowtail\mathcal{E}_{\hbar}\leftarrowtail\mathfrak{B}_{\alpha(\hbar)}
\end{align}
defined in Eqs. (\ref{eq:limvec})-(\ref{eq:qinner}) is a Hilbert bimodule.  In the case of $\hbar = 0_I$, we refer to $\mathcal{E}_{0_I}$ defined as above as \emph{the (Hilbert) classical limit of} $\mathcal{E}$.
\end{theorem}

\begin{proof}
\begin{enumerate}
\item First, we establish that the left action of $\mathfrak{A}_{\hbar}$ given in Eq. (\ref{eq:qAaction}) is well-defined.

Suppose that $\varphi,\varphi'\in\mathcal{E}$ are such that $\gamma(\varphi) = \gamma(\varphi')$ so that $\varphi-\varphi'\in\overline{\mathcal{E}\cdot K_{\alpha(\hbar)}}$.  Then for any $a\in \mathfrak{A}$, it follows from the fact that the $\mathfrak{A}$- and $\mathfrak{B}$-actions commute that $a\varphi - a\varphi'\in \overline{\mathcal{E}\cdot K_{\alpha(\hbar)}}$ so that $\gamma(a\varphi) = \gamma(a\varphi')$.

Likewise, suppose that $a,a'\in\mathfrak{A}$ are such that $\phi_{\hbar}(a) = \phi_{\hbar}(a')$ so that $a-a'\in K_{\hbar}$.  Then strong non-degeneracy implies that $a\varphi - a'\varphi = (a-a')\varphi \in \overline{\mathcal{E}\cdot K_{\alpha(\hbar)}}$.  It follows that $\gamma(a\varphi) = \gamma(a'\varphi)$.

\item Second, we establish that the right action of $\mathfrak{B}_{\alpha(\hbar)}$ given in Eq. (\ref{eq:qBaction}) is well-defined.

Suppose that $\varphi,\varphi'\in\mathcal{E}$ are such that $\gamma(\varphi) = \gamma(\varphi')$ so that $\varphi - \varphi' \in \overline{\mathcal{E}\cdot K_{\alpha(\hbar)}}$.  Then for any $b\in \mathfrak{B}$, it follows from the fact that $K_{\alpha(\hbar)}$ is an ideal that $\varphi\cdot b - \varphi'\cdot b\in \overline{\mathcal{E}\cdot K_{\alpha(\hbar)}}$ so that $\gamma(\varphi \cdot b) = \gamma(\varphi'\cdot b)$.

Likewise, suppose that $b,b'\in \mathfrak{B}$ are such that $\psi_{\alpha(\hbar)}(b) = \psi_{\alpha(\hbar)}(b')$ so that $b-b'\in K_{\alpha(\hbar)}$.  Then we have $\varphi\cdot b - \varphi\cdot b' = \varphi \cdot (b-b') \in \overline{\mathcal{E}\cdot K_{\alpha(\hbar)}}$.  It follows that $\gamma(\varphi \cdot b) = \gamma (\varphi\cdot b')$.

\item Third, we establish that the bilinear map given in Eq. (\ref{eq:qinner}) is well-defined.

Suppose that $\varphi,\varphi'\in \mathcal{E}$ are such that $\gamma(\varphi) = \gamma(\varphi')$ so that $\varphi - \varphi'\in \overline{\mathcal{E}\cdot K_{\alpha(\hbar)}}$.  Then since $K_{\alpha(\hbar)}$ is an ideal, we have $\inner{\varphi - \varphi'}{\xi}_{\mathfrak{B}}\in K_{\alpha(\hbar)}$ and $\inner{\xi}{\varphi - \varphi'}_{\mathfrak{B}}\in K_{\alpha(\hbar)}$ for any $\xi\in \mathcal{E}$.  Thus, we have
\begin{align}
    \inner{\gamma(\varphi)}{\gamma(\xi)}_{\mathfrak{B}_{\alpha(\hbar)}} - \inner{\gamma(\varphi')}{\gamma(\xi)}_{\mathfrak{B}_{\alpha(\hbar)}} &= \psi_{\alpha(\hbar)}\big(\inner{\varphi}{\xi}_{\mathfrak{B}}\big) - \psi_{\alpha(\hbar)}\big(\inner{\varphi'}{\xi}_{\mathfrak{B}}\big)\nonumber\\
    &= \psi_{\alpha(\hbar)}\big(\inner{\varphi- \varphi'}{\xi}_{\mathfrak{B}}\big) = 0
\end{align}
so that $\inner{\gamma(\varphi)}{\gamma(\xi)}_{\mathfrak{B}_{\alpha(\hbar)}} = \inner{\gamma(\varphi')}{\gamma(\xi)}_{\mathfrak{B}_{\alpha(\hbar)}}$.  Likewise, we have
\begin{align}
    \inner{\gamma(\xi)}{\gamma(\varphi)}_{\mathfrak{B}_{\alpha(\hbar)}} - \inner{\gamma(\xi)}{(\gamma(\varphi')}_{\mathfrak{B}_{\alpha(\hbar)}} &=\psi_{\alpha(\hbar)}\big(\inner{\xi}{\varphi}_{\mathfrak{B}}\big) - \psi_{\alpha(\hbar)}\big(\inner{\xi}{\varphi'}_{\mathfrak{B}}\big)\nonumber\\
    &=\psi_{\alpha(\hbar)}\big(\inner{\xi}{\varphi-\varphi'}_{\mathfrak{B}}\big) = 0
\end{align}
so that $\inner{\gamma(\xi)}{\gamma(\varphi)}_{\mathfrak{B}_{\alpha(\hbar)}} = \inner{\gamma(\xi)}{(\gamma(\varphi')}_{\mathfrak{B}_{\alpha(\hbar)}}$.

\item Fourth, we establish that the bilinear map given in Eq. (\ref{eq:qinner}) is a positive semi-definite inner product.

Suppose $\varphi,\xi\in\mathcal{E}$.  Then
\begin{align}
    \inner{\gamma(\varphi)}{\gamma(\xi)}_{\mathfrak{B}_{\alpha(\hbar)}}^* = \psi_{\alpha(\hbar)}\big(\inner{\varphi}{\xi}_{\mathfrak{B}}\big)^* = \psi_{\alpha(\hbar)}\big(\inner{\xi}{\varphi}_{\mathfrak{B}}\big) = \inner{\gamma(\xi)}{\gamma(\varphi)}_{\mathfrak{B}_{\alpha(\hbar)}},
\end{align}
so the inner product is Hermitian.  For any $b\in\mathfrak{B}$, we have
\begin{align}
    \inner{\gamma(\varphi)}{\gamma(\xi)\cdot \psi_{\alpha(\hbar)}(b)}_{\mathfrak{B}_{\alpha(\hbar)}} = \psi_{\alpha(\hbar)}\big(\inner{\varphi}{\xi\cdot b}_{\mathfrak{B}}\big) = \psi_{\alpha(\hbar)}\big(\inner{\varphi}{\xi}_{\mathfrak{B}}\cdot b\big) =  \inner{\gamma(\varphi)}{\gamma(\xi)}\cdot \psi_{\alpha(\hbar)}(b),
\end{align}
so the inner product is compatible with the right $\mathfrak{B}_{\alpha(\hbar)}$-action.  And finally, we have 
\begin{align}
    \inner{\gamma(\varphi)}{\gamma(\varphi)}_{\mathfrak{B}_{\alpha(\hbar)}} = \psi_{\alpha(\hbar)}\big(\inner{\varphi}{\varphi}_{\mathfrak{B}}\big)\geq 0,
\end{align}
so the inner product is positive semi-definite.

\item Fifth, we establish that the vector space $\mathcal{E}_{\hbar}$ is complete in the norm induced by the inner product, and that the inner product is positive definite.

Our strategy is to show that the norm induced by the inner product coincides with the quotient norm on $\mathcal{E}_{\hbar}$, in which the space is already complete and which is known to be positive definite.

We employ results about the linking algebra $\mathfrak{L}$ and its two-sided ideal $\mathfrak{D}$ established in Appendix \ref{app:link}.  We have that for any $\varphi\in \mathcal{E}$,
\begin{align}
    \norm{\gamma(\varphi)}_{\mathfrak{B}_{\alpha(\hbar)}} &= \left\Vert\begin{pmatrix}
    0 & \gamma(\varphi)\nonumber\\
    0 & 0
    \end{pmatrix}\right\Vert_{\mathfrak{L}/\mathfrak{D}}\nonumber\\
    &= \inf_{D\in\mathfrak{D}} \left\Vert \begin{pmatrix}
    0 & \varphi\nonumber\\
    0 & 0
    \end{pmatrix} + D\right\Vert\nonumber\\
    &= \inf_{\xi\in\overline{\mathcal{E}\cdot K_{\alpha(\hbar)}}} \left\Vert \begin{pmatrix}
    0 & \varphi\nonumber\\
    0 & 0
    \end{pmatrix} + \begin{pmatrix}
    0 & \xi\nonumber\\
    0 & 0
    \end{pmatrix}\right\Vert\nonumber\\
    &= \inf_{\xi\in \overline{\mathcal{E}\cdot K_{\alpha(\hbar)}}} \left\Vert \varphi + \xi\right\Vert.
\end{align}
The first, third, and fourth equalities follow from the bounds established in Lemma \ref{lemma:link}.  Since the last line is the canonical quotient norm on the vector space $\mathcal{E}_{\hbar}$, which is positive definite, and since the space $\mathcal{E}_{\hbar}$ is complete relative to this norm, this completes the proof.
\end{enumerate}
\end{proof}

Thus, we have shown that from a continuous proper map $\alpha: I\to J$ and a Hilbert bimodule $\mathfrak{A}\rightarrowtail\mathcal{E}\leftarrowtail\mathfrak{B}$ between the algebras of uniformly continuous sections of the bundles that is strongly non-degenerate for $\alpha$ at $\hbar = 0_I$, one can naturally define a Hilbert bimodule between the fibers at $\hbar = 0_I$ and $\hbar = 0_J$ given by the Hilbert classical limit
\begin{align}
\mathfrak{A}_{0_I}\rightarrowtail\mathcal{E}_{0_I}\leftarrowtail\mathfrak{B}_{0_J}.
\end{align}
It may be helpful to illustrate the Hilbert classical limit with our running example.\medskip

\begin{example}
    Suppose again we are in the situation of Ex. \ref{ex:bimodfam}. 
 That is, suppose $\beta: \mathfrak{A}\to\mathfrak{B}$ is a non-degenerate *-homomorphism between the C*-algebras of uniformly continuous sections satisfying $\beta(K_\hbar)\subseteq K_{\alpha(\hbar)}$ for each $\hbar\in I$.  One can take the classical limit of a *-homomorphism directly by the construction of \citet[][p. 14, Prop. 5.2]{StFe21a} by factoring through the quotients on $\mathfrak{A}_{0_I} = \mathfrak{A}/K_{0_I}$ and $\mathfrak{B}_{0_J} = \mathfrak{B}/K_{0_J}$ to obtain the unique non-degenerate *-homomorphism $\beta_{0_I}:\mathfrak{A}_{0_I}\to\mathfrak{B}_{0_J}$ satisfying $\beta_{0_I}\circ \phi_{0_I} = \psi_{0_I}\circ \beta$.  Hence, $\beta_{0_I}$ is associated with a Hilbert bimodule
 \begin{align}
\label{eq:classbimod}\mathfrak{A}_{0_I}\rightarrowtail\mathfrak{B}_{0_J}\leftarrowtail\mathfrak{B}_{0_J}
 \end{align}
 by the construction in Ex. \ref{ex:*hombimod}.

 On the other hand, we have that $\beta$ is directly associated with a Hilbert bimodule
 \begin{align}     \mathfrak{A}\rightarrowtail\mathfrak{B}\leftarrowtail\mathfrak{B},
 \end{align}
 by the construction in Ex. \ref{ex:*hombimod}.  This Hilbert bimodule is strongly non-degenerate at $\hbar = 0_I\in \dot{I}$, so it has a Hilbert classical limit as determined in Thm. \ref{thm:quotient}.  Since $\mathfrak{B}/\overline{\mathfrak{B}\cdot K_{\alpha(\hbar)}} = \mathfrak{B}/K_{\alpha(\hbar)} = \mathfrak{B}_{\alpha(\hbar)}$ for each $\hbar\in \dot{I}$, we have that the Hilbert classical limit at $\hbar = 0_I$ is the bimodule
 \begin{align} \mathfrak{A}_{0_I}\rightarrowtail\mathfrak{B}_{0_J}\leftarrowtail\mathfrak{B}_{0_J},
 \end{align}
 which is equivalent to the bimodule in Eq. (\ref{eq:classbimod}).
 
 Hence, the classical limit of a Hilbert bimodule determined by a *-homomorphism agrees with the Hilbert bimodule determined by the classical limit of a *-homomorphism.  Thus, the Hilbert classical limit of Hilbert bimodules we have defined provides an extension of the previous notion of the classical limit of a *-homomorphism from \citet[][p. 14, Lemma 4.3]{StFe21a}.
\end{example}

\subsection{Functoriality of the Quotient Bimodule}

Next, we will next establish that the correspondence given by the classical limit is functorial in the sense that it respects composition of bimodules by the tensor product.  As in Thm. \ref{thm:quotient} of the previous section, we suppose $\mathfrak{A}$ and $\mathfrak{B}$ are the algebras of uniformly continuous sections of uniformly continuous bundles of C*-algebras, and we suppose $\mathfrak{A}\rightarrowtail\mathcal{E}\leftarrowtail\mathfrak{B}$ is a strongly non-degenerate Hilbert bimodule for $\alpha$ at $0_I\in I$.  This implies that $\mathcal{E}$ has a Hilbert classical limit, which we denote $\mathfrak{A}_{0_I}\rightarrowtail\mathcal{E}_{0_I}\leftarrowtail\mathfrak{B}_{0_J}$.

 Moreover, we suppose now that $((\mathfrak{C}_\hbar,\zeta_\hbar)_{\hbar\in L},\mathfrak{C})$ is a further uniformly continuous bundle of C*-algebras over a locally compact metric space $L$, whose one-point compactification $\dot{L} = L\cup\{0_L\}$ is likewise a metric space with isometric inclusion $L\hookrightarrow \dot{L}$.  We use the notation
\begin{align}
    K_{0_L} = \{c\in\mathfrak{C}\ |\ \lim_{\hbar\to 0_L}\norm{\zeta_\hbar(c)} = 0\}
\end{align}
and we denote the canonical quotient mapping by
\begin{align}
    \zeta_{0_L}: \mathfrak{C}\to \mathfrak{C}_{0_L}
\end{align}
with $\mathfrak{C}_{0_L} = \mathfrak{C}/K_{0_L}$.  We now denote the map $\alpha:I\to J$ from the previous section by $\alpha_{IJ}$.  We suppose further that $\alpha_{JL}: J\to L$ is a continuous proper map and $\mathfrak{B}\rightarrowtail \mathcal{F}\leftarrowtail\mathfrak{C}$ is a strongly non-degenerate Hilbert bimodule for $\alpha_{JL}$ at $\hbar = 0_J$.  We consider the bimodule
\begin{align}
\mathfrak{A}_{0_I}\rightarrowtail(\mathcal{E}\otimes_{\mathfrak{B}}\mathcal{F})/\overline{\mathcal{E}\otimes_{\mathfrak{B}}\mathcal{F}\cdot K_{0_L}}\leftarrowtail \mathfrak{C}_{0_L}
\end{align}
which is obtained by first taking the interior tensor product of the bimodules and afterwards taking the Hilbert classical limit.  We must compare this to the bimodule
\begin{align} \mathfrak{A}_{0_I}\rightarrowtail (\mathcal{E}/\overline{\mathcal{E}\cdot K_{0_J}})\otimes_{\mathfrak{B}_{0_J}}(\mathcal{F}/\overline{\mathcal{F}\cdot K_{0_L}})\leftarrowtail \mathfrak{C}_{0_L},
\end{align}
which is obtained by first taking the Hilbert classical limit of the bimodules and afterward taking the interior tensor product.  We will now show that these bimodules are unitarily equivalent.\medskip

\begin{theorem}
Suppose that $\alpha_{IJ}: I\to J$ and $\alpha_{JL}: J\to L$ are continuous proper maps, and that $\mathfrak{A}\rightarrowtail\mathcal{E}\leftarrowtail\mathfrak{B}$ and $\mathfrak{B}\rightarrowtail\mathcal{F}\leftarrowtail\mathfrak{C}$ are strongly non-degenerate Hilbert bimodules for $\alpha$ and $\beta$ at $\hbar = 0_I$ and $\hbar = 0_J$.  Then the bimodules 
\begin{align}
    \mathfrak{A}_{0_I}\rightarrowtail (\mathcal{E}\otimes_{\mathfrak{B}}\mathcal{F})/\overline{\mathcal{E}\otimes_{\mathfrak{B}}\mathcal{F}\cdot K_{0_L}}\leftarrowtail \mathfrak{C}_{0_L},
\end{align}
and 
\begin{align}
    \mathfrak{A}_{0_I}\rightarrowtail (\mathcal{E}/\overline{\mathcal{E}\cdot K_{0_J}})\otimes_{\mathfrak{B}_{0_J}}(\mathcal{F}/\overline{\mathcal{F}\cdot K_{0_L}})\leftarrowtail \mathfrak{C}_{0_L},
\end{align}
are unitarily equivalent.
\end{theorem}

\begin{proof}
We will explicitly define unitary maps
\begin{align}
    U: (\mathcal{E}\otimes_{\mathfrak{B}}\mathcal{F})/\overline{\mathcal{E}\otimes_{\mathfrak{B}}\mathcal{F}\cdot K_{0_L}}\to (\mathcal{E}/\overline{\mathcal{E}\cdot K_{0_J}})\otimes_{\mathfrak{B}_{0_J}}(\mathcal{F}/\overline{\mathcal{F}\cdot K_{0_L}})\nonumber\\
    U^*: (\mathcal{E}/\overline{\mathcal{E}\cdot K_{0_J}})\otimes_{\mathfrak{B}_{0_J}}(\mathcal{F}/\overline{\mathcal{F}\cdot K_{0_L}})\to (\mathcal{E}\otimes_{\mathfrak{B}}\mathcal{F})/\overline{\mathcal{E}\otimes_{\mathfrak{B}}\mathcal{F}\cdot K_{0_L}}.
\end{align}
by the continuous linear extension of the assignments
\begin{align}
U&\Big(\varphi \otimes_{\mathfrak{B}} \psi  + \overline{\mathcal{E}\otimes_{\mathfrak{B}}\mathcal{F}\cdot K_{0_L}}\Big) = (\varphi + \overline{\mathcal{E}\cdot K_{0_J}})\otimes_{\mathfrak{B}_{0_J}} (\psi + \overline{\mathcal{F}\cdot K_{0_L}})\nonumber\\
U^*&\Big((\varphi + \overline{\mathcal{E}\cdot K_{0_J}})\otimes_{\mathfrak{B}_{0_J}} (\psi + \overline{\mathcal{F}\cdot K_{0_L}}) \Big) = \varphi\otimes_{\mathfrak{B}}\psi + \overline{\mathcal{E}\otimes_{\mathfrak{B}}\mathcal{F}\cdot K_{0_L}}
\end{align}
for all $\varphi\in\mathcal{E}$ and $\psi\in\mathcal{F}$.

First, note that these maps are well-defined.  Suppose we  have $\varphi,\varphi'\in\mathcal{E}$ with $\varphi - \varphi' = \sum_i\xi_i\cdot b_i$ for some $b_i\in K_{0_J}$ and $\xi_i\in \mathcal{E}$.  Suppose also we have $\psi,\psi'\in\mathcal{F}$ with $\psi - \psi' = \sum_j\eta_j\cdot c_j$ for some $c_j\in K_{0_L}$ and $\eta_j\in \mathcal{F}$.  It follows from strong non-degeneracy that $\sum_i b_i\psi' = \sum_k \chi_k c'_k$ for some $\chi_k\in \mathcal{F}$ and $c'_k\in K_{0_L}$.  Hence,
\begin{align}
\varphi &\otimes_{\mathfrak{B}_{0_J}} \psi = \left(\varphi' + \sum_i\xi_i\cdot b_i\right)\otimes_{\mathfrak{B}_{0_J}} \left(\psi' + \sum_j\eta_j\cdot c_j\right)\nonumber\\
&= \varphi'\otimes_{\mathfrak{B}_{0_J}} \psi' + \left(\varphi' + \sum_i\xi_i\cdot b_i\right)\otimes _{\mathfrak{B}_{0_J}}\left(\sum_j\eta_j\cdot c_j\right) + \left(\sum_i\xi_i\cdot b_i\right) \otimes_{\mathfrak{B}_{0_J}} \psi'\nonumber\\
&= \varphi'\otimes_{\mathfrak{B}_{0_J}} \psi' + \left(\varphi' + \sum_i\xi_i\cdot b_i\right)\otimes _{\mathfrak{B}_{0_J}}\left(\sum_j\eta_j\cdot c_j\right) + \sum_i\xi_i\otimes _{\mathfrak{B}_{0_J}}\left(b_i\cdot \psi'\right)\nonumber\\
&=\varphi'\otimes_{\mathfrak{B}_{0_J}} \psi' + \left(\varphi' + \sum_i\xi_i\cdot b_i\right)\otimes _{\mathfrak{B}_{0_J}}\left(\sum_j\eta_j\cdot c_j\right) + \sum_i\xi_i\otimes_{\mathfrak{B}_{0_J}}\left(\sum_k\chi_k\cdot c'_k\right).
\end{align}

Note that
\begin{align}
    \left(\varphi' + \sum_i\xi_i\cdot b_i\right)\otimes _{\mathfrak{B}_{0_J}}\left(\sum_j\eta_j\cdot c_j\right)&= 0\nonumber\\
    \sum_i\xi_i\otimes_{\mathfrak{B}_{0_J}}\left(\sum_k\chi_k\cdot c'_k\right) &= 0
    \end{align}
in $(\mathcal{E}/\overline{\mathcal{E}\cdot K_{0_J}})\otimes_{\mathfrak{B}_{0_J}}(\mathcal{F}/\overline{\mathcal{F}\cdot K_{0_L}})$ so that
\begin{align}
U\Big(\varphi\otimes_{\mathfrak{B}}\psi + \overline{\mathcal{E}\otimes_{\mathfrak{B}}\mathcal{F}\cdot K_{0_L}}\Big) &= U\Big(\varphi'\otimes_{\mathfrak{B}}\psi'+ \overline{\mathcal{E}\otimes_{\mathfrak{B}}\mathcal{F}\cdot K_{0_L}}\Big)\nonumber\\
U^*\Big((\varphi + \overline{\mathcal{E}\cdot K_{0_J}})\otimes_{\mathfrak{B}_{0_J}}(\psi + \overline{\mathcal{F}\cdot K_{0_L}})\Big) &= U^*\Big((\varphi' + \overline{\mathcal{E}\cdot K_{0_J}})\otimes_{\mathfrak{B}_{0_J}}(\psi' + \overline{\mathcal{F}\cdot K_{0_L}})\Big).
\end{align}

Next, we show that $U$ and $U^*$ are unitary.  We denote the quotient maps for the vector spaces by
\begin{align}
\gamma_I&: \mathcal{E}\to \mathcal{E}/\overline{\mathcal{E}\cdot K_{0_J}}\nonumber\\
\gamma_J&: \mathcal{F}\to\mathcal{F}/\overline{\mathcal{F}\cdot K_{0_L}}\nonumber\\
\gamma_{IJ}&: \mathcal{E}\otimes_{\mathfrak{B}}\mathcal{F}\to (\mathcal{E}\otimes_{\mathfrak{B}}\mathcal{F})/\overline{\mathcal{E}\otimes_{\mathfrak{B}}\mathcal{F}\cdot K_{0_L}}.
\end{align}
Clearly, from what has already been shown, we have that for any $\varphi\in\mathcal{E}$ and $\psi\in\mathcal{F}$ that $U\big(\gamma_{IJ}(\varphi\otimes_{\mathfrak{B}}\psi)\big) = \gamma_I(\varphi)\otimes_{\mathfrak{B}_{0_J}}\gamma_J(\psi)$.

Then for any $\varphi,\varphi'\in\mathcal{E}$ and $\psi,\psi'\in\mathcal{F}$, we have
\begin{align}
\left \langle U(\gamma_{IJ}(\varphi\otimes_{\mathfrak{B}}\psi), U(\gamma_{IJ}(\varphi'\otimes_{\mathfrak{B}}\psi'))\right \rangle_{\mathfrak{C}_{0_L}} & = \left \langle\gamma_J(\psi), \inner{\gamma_I(\varphi)}{\gamma_I(\varphi')}_{\mathfrak{B}_{0_J}}\cdot \gamma_J(\psi')\right \rangle_{\mathfrak{C}_{0_L}}\nonumber\\
&= \left \langle \gamma_J(\psi), \psi_{0_J}\big(\inner{\varphi}{\varphi'}_{\mathfrak{B}}\big)\cdot \gamma_J(\psi')\right \rangle_{\mathfrak{C}_{0_L}}\nonumber\\
& = \left \langle \gamma_J(\psi), \gamma_J\big(\inner{\varphi}{\varphi'}_{\mathfrak{B}}\cdot \psi'\big)\right \rangle_{\mathfrak{C}_{0_L}}\nonumber\\
& = \zeta_{0_L}\Big(\left \langle \psi, \inner{\varphi}{\varphi'}_{\mathfrak{B}}\cdot \psi'\right \rangle_{\mathfrak{C}}\Big)\nonumber\\
&= \zeta_{0_L}\Big( \inner{\varphi\otimes_{\mathfrak{B}}\psi}{\varphi'\otimes_{\mathfrak{B}}{\psi'}}_{\mathfrak{C}}\Big)\nonumber\\
&= \left \langle \gamma_{IJ}(\varphi\otimes_{\mathfrak{B}}\psi), \gamma_{IJ}(\varphi'\otimes_{\mathfrak{B}}\psi')\right\rangle_{\mathfrak{C}_{0_L}}
\end{align}
and
\begin{align}
\left\langle U^*\big(\gamma_I(\varphi)\otimes_{\mathfrak{B}_{0_J}}\gamma_J(\psi)\big), U^*\big(\gamma_I(\varphi')\otimes_{\mathfrak{B}_{0_J}}\gamma_J(\psi')\big)\right \rangle_{\mathfrak{C}_{0_L}} &= \left\langle \gamma_{IJ}(\varphi\otimes_{\mathfrak{B}}\psi), \gamma_{IJ}(\varphi'\otimes_{\mathfrak{B}}\psi')\right\rangle_{\mathfrak{C}_{0_L}}\nonumber\\
&= \zeta_{0_L}\Big ( \left \langle \varphi\otimes_{\mathfrak{B}}\psi, \varphi'\otimes_{\mathfrak{B}}\psi' \right \rangle_{\mathfrak{C}} \Big)\nonumber\\
&= \zeta_{0_L}\Big( \left \langle \psi, \inner{\varphi}{\varphi'}_{\mathfrak{B}}\cdot \psi'\right \rangle_{\mathfrak{C}} \Big)\nonumber\\
& = \left \langle \gamma_{J}(\psi), \gamma_J\big( \inner{\varphi}{\varphi'}_{\mathfrak{B}}\cdot \psi'\big) \right \rangle_{\mathfrak{C}_{0_L}}\nonumber\\
& = \left \langle \gamma_{J}(\psi), \psi_{0_J}\big( \inner{\varphi}{\varphi'}_{\mathfrak{B}}\big)\cdot \gamma_J(\psi') \right \rangle_{\mathfrak{C}_{0_L}}\nonumber\\
&= \left \langle \gamma_{J}(\psi), \inner{\gamma_I(\varphi)}{\gamma_I(\varphi')}_{\mathfrak{B}_{0_J}}\cdot \gamma_J(\psi') \right \rangle_{\mathfrak{C}_{0_L}}\nonumber\\
&= \left \langle \gamma_I(\varphi)\otimes_{\mathfrak{B}_{0_J}}\gamma_J(\psi), \gamma_I(\varphi')\otimes_{\mathfrak{B}_{0_J}}\gamma_J(\psi')\right \rangle_{\mathfrak{C}_{0_L}}.
\end{align}
Moreover, $U$ and $U^*$ are clearly inverses, which completes the proof.
\end{proof}

This establishes the functoriality of the classical limit Hilbert bimodule.  Next, we show that once one has a classical limit Hilbert bimodule, one can construct a corresponding symplectic dual pair.

\section{Symplectic Dual Pair of a Classical Hilbert Bimodule}
\label{sec:dual}

We now fix the base space for simplicity by supposing that $I=J= (0,1]$.  Suppose that we have two strict deformation quantizations of Poisson manifolds $M$ and $N$.  We will denote by $(\mathfrak{A}_\hbar,\mathcal{Q}_\hbar)_{\hbar\in (0,1]}$ the strict deformation quantization of $\mathcal{P}_M = C_c^\infty(M)$, and we denote by $(\mathfrak{B}_\hbar,\mathcal{Q}'_\hbar)_{\hbar\in (0,1]}$ the strict deformation quantization of $\mathcal{P}_N = C_c^\infty(N)$.  We suppose that each of these strict deformation quantizations generates a corresponding uniformly continuous bundle of C*-algebras, denoted respectively by $((\mathfrak{A}_\hbar,\phi_\hbar)_{\hbar\in (0,1]},\mathfrak{A})$ and $((\mathfrak{B}_\hbar,\psi_\hbar)_{\hbar\in (0,1]},\mathfrak{B})$.  Suppose further that we are given a Hilbert bimodule
\begin{align}
\mathfrak{A}\rightarrowtail\mathcal{E}\leftarrowtail\mathfrak{B}
\end{align}
that is strongly non-degenerate at every $\hbar\in [0,1]$.  Then Thm. \ref{thm:quotient} establishes that $\mathcal{E}$ has a Hilbert classical limit, which we denote by
\begin{align}
\mathfrak{A}_0\rightarrowtail\mathcal{E}_0\leftarrowtail\mathfrak{B}_0.
\end{align}
We now proceed to construct a symplectic dual pair from this Hilbert bimodule $\mathcal{E}_0$ between the abelian C*-algebras $\mathfrak{A}_0$ and $\mathfrak{B}_0$.

\subsection{Construction of a Symplectic Space}
\label{sec:symplim}

We begin with a few preliminaries.  Recall that since $\mathcal{E}$ is assumed to be strongly non-degenerate at each $\hbar\in [0,1]$, it determines through Thm. \ref{thm:quotient} not only a Hilbert classical limit $\mathcal{E}_0$, but also a Hilbert bimodule $\mathcal{E}_\hbar$ between the fiber C*-algebras $\mathfrak{A}_\hbar$ and $\mathfrak{B}_\hbar$ at each value $\hbar\in[0,1]$.  It will be useful for us to rewrite the left $\mathfrak{A}_\hbar$-action and right $\mathfrak{B}_\hbar$-action on $\mathcal{E}_\hbar$ in terms of representations $\pi^{\mathcal{E}_\hbar}_{\mathfrak{A_\hbar}}: \mathfrak{A}_\hbar\to \mathcal{B}(\mathcal{E}_\hbar)$ and $\pi^\mathcal{E_\hbar}_{\mathfrak{B_\hbar}}: \mathfrak{B}_\hbar\to\mathcal{B}(\mathcal{E}_\hbar)$ as bounded operators on the Banach space $\mathcal{E}_\hbar$ defined by
\begin{align}
\pi^{\mathcal{E}_\hbar}_{\mathfrak{A}_\hbar}(a)\varphi &= a\cdot \varphi\nonumber\\
\pi^{\mathcal{E}_\hbar}_{\mathfrak{B}_\hbar}(b)\varphi &= \varphi\cdot b
\end{align}
for all $a\in\mathfrak{A}_\hbar$, $b\in\mathfrak{B}_\hbar$, and $\varphi\in\mathcal{E}_\hbar$.  The representations $\pi^{\mathcal{E}_\hbar}_{\mathfrak{A}_\hbar}$ and $\pi^{\mathcal{E}_\hbar}_{\mathfrak{B}_\hbar}$ are clearly defined for every value $\hbar\in [0,1]$.  Now, consider the tensor product C*-algebra $\mathfrak{A}_\hbar\otimes\mathfrak{B}_\hbar$, defined as the completion of the algebraic tensor product in the maximal C*-norm.  (In what follows, we always employ the maximal C*-norm for tensor product C*-algebras.)  The fact that $\pi_{\mathfrak{A}_\hbar}^{\mathcal{E}_\hbar}(\mathfrak{A}_\hbar)$ commutes with $\pi_{\mathfrak{B}_\hbar}^{\mathcal{E}_\hbar}(\mathfrak{B}_\hbar)$ ensures that the representation $\pi^{\mathcal{E}_\hbar}_{\mathfrak{A}_\hbar}\otimes \pi^{\mathcal{E}_\hbar}_{\mathfrak{B}_\hbar}$ of the tensor product algebra $\mathfrak{A}_\hbar\otimes\mathfrak{B}_\hbar$ is a *-homomorphism.  The tensor product algebra and this tensor product representation will be important in what follows.

Further, we denote the Gelfand duals (i.e., pure state spaces) of $\mathfrak{A}$ and $\mathfrak{B}$ by
\begin{align}
\label{eq:gelfand}
M &\cong \mathcal{P}(\mathfrak{A}_0)\nonumber\\
N &\cong \mathcal{P}(\mathfrak{B}_0).
\end{align}
The Gelfand dual of the C*-tensor product $\mathfrak{A}_0\otimes\mathfrak{B}_0$ is the product space
\begin{align}
\mathcal{P}(\mathfrak{A}_0\otimes\mathfrak{B}_0) \cong M\times N.
\end{align}
This is well known when $\mathfrak{A}_0$ and $\mathfrak{B}_0$ are unital C*-algebras, but it also generalizes to non-unital C*-algebras, as established in Appendix \ref{app:prod}.

Recall again from \S\ref{sec:back} that following \citet[][p. 34-35]{Ne03}, we define the smooth envelope of an algebra $\mathcal{P}$ of functions on a space $X$ as the algebra of all functions on the same space of the form
\begin{align}
g(f_1,...,f_n)(x) = g(f_1(x),...,f_n(x))
\end{align}
for all $x\in X$, where $f_1,...,f_n\in\mathcal{P}$ and $g\in C^\infty(\R^n)$.  We denote the smooth envelope of $\mathcal{P}$ by $\overline{\mathcal{P}}$.  Then for the smooth manifolds $M$ and $N$, we have $\overline{C_c^\infty(M)}\cong C^\infty(M)$ and $\overline{C_c^\infty(N)}\cong C^\infty(N)$.

This allows us to define the following structures:
\begin{itemize}
\item Define the C*-algebra
\begin{align}
\label{eq:dualpairalgebra}
\mathfrak{A}_0\circledast_\mathcal{E}\mathfrak{B}_0 = (\mathfrak{A}_0\otimes\mathfrak{B}_0)/\ker(\pi^{\mathcal{E}_0}_{\mathfrak{A}_0}\otimes\pi^{\mathcal{E}_0}_{\mathfrak{B}_0})
\end{align}
with canonical quotient map
\begin{align}
\rho: \mathfrak{A}_0\otimes\mathfrak{B}_0\to\mathfrak{A}_0\circledast_\mathcal{E}\mathfrak{B}_0.
\end{align}
We define a topological space $S$ as the Gelfand dual
\begin{align}
\label{eq:dualpairspace}
S = \mathcal{P}\big(\mathfrak{A}_0\circledast_\mathcal{E}\mathfrak{B}_0\big)
\end{align}
with the weak* topology.  Hence, we also have a continuous dual map
\begin{align}
\hat{\rho}: S \to M\times N.
\end{align}
Since $\rho$ is a surjection, it follows that $\hat{\rho}$ is an injection.
\item Define the algebra
\begin{align}
\label{eq:smooth}
\mathcal{P}_S = \hat{\rho}^*\Big[\overline{\overline{\mathcal{P}_M}\otimes \overline{\mathcal{P}_N}}\Big].
\end{align}
Note that this is well-defined because $\overline{\overline{\mathcal{P}_M}\otimes\overline{\mathcal{P}_N}}\cong C^\infty(M\times N)$  \citep[][Prop. 4.29, p. 49]{Ne03}.

In what follows, we suppose that $\mathcal{P}_S$ satisfies the known conditions provided by Nestruev of being smooth \citep[][p. 37]{Ne03}, geometric \citep[][p. 23]{Ne03}, and complete \citep[][p. 31]{Ne03}.  As mentioned in \S\ref{sec:back}, Nestruev shows that these are necessary and sufficient conditions to guarantee that $S$ is a smooth manifold  \citep[][p. 77, Thm. 7.2 and p. 79, Thm. 7.7]{Ne03}.  By construction, $\hat{\rho}$ is then a smooth map between manifolds.

\item  We define maps $J_M: S\to M$ and $J_N: S\to N$ given by
\begin{align}
\label{eq:proj}
J_M(s)(a) &= \lim_\delta (s\circ\rho)(a\otimes I_\delta^\mathfrak{B})\nonumber\\
J_N(s)(b)&= \lim_\delta (s\circ\rho)(I_\delta^\mathfrak{A}\otimes b)
\end{align}
for all $s\in S$, $a\in\mathfrak{A}$, and $b\in\mathfrak{B}$, where $(I^\mathfrak{A}_\delta)_{\delta\in \Delta}$ and $(I^\mathfrak{B}_\delta)_{\delta\in \Delta'}$ are approximate identities for $\mathfrak{A}_0$ and $\mathfrak{B}_0$, respectively.  Denoting the smooth projections of the product manifold by $pr_M: M\times N\to M$ and $pr_N: M\times N\to N$, we have $J_M = pr_M\circ \hat{\rho}$ and $J_N = pr_N\circ\hat{\rho}$.  It follows that $J_M$ and $J_N$ are smooth maps.

\item The Poisson bracket on $\mathcal{P}_S$ is defined in three steps.  Recall that the Poisson brackets on each of $\mathcal{P}_M$ and $\mathcal{P}_N$, denoted by $\{\cdot,\cdot\}_M$ and $\{\cdot,\cdot\}_N$, respectively, are determined from the strict quantization by Dirac's condition in Def. \ref{def:strquant}.

We proceed from this starting point.

\begin{enumerate}
\item First, we define Poisson brackets on each of $\overline{\mathcal{P}_M}$  and $\overline{\mathcal{P}_N}$ by extending the Poisson brackets locally on each of $M$ and $N$.  Suppose we are given arbitrary $\tilde{a}_1,\tilde{a}_2\in \overline{\mathcal{P}_M} = C^\infty(M)$ and $\tilde{b}_1,\tilde{b}_2\in\overline{\mathcal{P}_N} = C^\infty(N)$.  Then for any points $m\in M$ and $n\in N$, we know that there are neighborhoods $U$ of $m$ and $V$ of $n$ and compactly supported functions $a_i\in \mathcal{P}_M = C_c^\infty(M)$ and $b_j\in\mathcal{P}_N = C_c^\infty(N)$ such that $(\tilde{a}_i)_{|U} = (a_i)_{|U}$ and $(\tilde{b}_j)_{|V} = (b_j)_{|V}$.  So we can define
\begin{align}
\label{eq:poiss}
\{\tilde{a}_1,\tilde{a}_2\}_M(m) &= \{a_1,a_2\}_M(m)\nonumber\\
\{\tilde{b}_1,\tilde{b}_2\}_N(n) &= \{b_1,b_2\}_N(n).
\end{align}

\item Second, we define a bracket on $\overline{\mathcal{P}_M}\otimes\overline{\mathcal{P}_N}$ by the standard construction of a product Poisson structure.  For any $a_1,a_2\in\overline{\mathcal{P}_M}$ and $b_1,b_2\in\overline{\mathcal{P}_N}$, we define the Poisson bracket as the linear extension of the assignment
\begin{align}
\big\{a_1\otimes b_1,a_2\otimes b_2\big\}_{M\times N} = \{a_1,a_2\}_M \otimes b_1b_2 + a_1a_2 \otimes \{b_1,b_2\}_N.
\end{align}

\item Third, we extend the Poisson bracket to the smooth envelope $\overline{\overline{\mathcal{P}_M}\otimes\overline{\mathcal{P}_N}}\cong C^\infty(M\times N)$ so that the Poisson bracket on $\mathcal{P}_S$ may be defined by the pull-back via $\hat{\rho}$.  Consider any two functions
\begin{align}
f = g(f_1,...,f_k) && f' = g'(f'_1,...f'_{k'})
\end{align}
in the smooth envelope $\overline{\overline{\mathcal{P}_M}\otimes\overline{\mathcal{P}_N}}$, i.e., we have $f_1,...,f_k, f'_1,...,f'_{k'}\in \overline{\mathcal{P}_M}\otimes\overline{\mathcal{P}_N}$, $g\in C^\infty(\R^k)$, and $g'\in C^\infty(\R^{k'})$.  Now, we can use the chain rule as a stipulation that fully defines the Poisson bracket of $f$ and $f'$. 
 Define for each $x\in M\times N$,
\begin{align}
\{f,f'\}_{M\times N}(x) = \sum_{i=1}^k\sum_{j=1}^{k'} \frac{\partial g}{\partial x_i}_{|(f_1(x),...,f_k(x))} \frac{\partial g'}{\partial x_j}_{|(f'_1(x),...,f'_{k'}(x))} \{f_i,f'_j\}_{M\times N}(x).
\end{align}
It is straightforward to check that this extension of the Poisson bracket is anti-symmetric and bi-linear, and that it satisfies the Leibniz rule and the Jacobi identity.  Finally, we define the Poisson bracket on $S$ by
\begin{align}
\label{eq:prodpoiss}
\{f\circ\hat{\rho},f'\circ\hat{\rho}\}_S = \{f,f'\}_{M\times N}\circ\hat{\rho}
\end{align}
for all $f,f'\in C^\infty(M\times N)$
\end{enumerate}
\end{itemize}

We will need one further condition to ensure that $S$ becomes a symplectic manifold with the Poisson structure just defined.  In order to state the condition, we extend the evaluation maps of the uniformly continuous bundles of C*-algebras to the C*-tensor product by defining
$\Phi_\hbar: \mathfrak{A}\otimes\mathfrak{B}\to \mathfrak{A}_\hbar\otimes\mathfrak{B}_\hbar$ as the continuous linear extension of
\begin{align}
\Phi_\hbar = \phi_\hbar\otimes\psi_\hbar
\end{align}
for all $a\in\mathfrak{A}$ and $b\in\mathfrak{B}$ at each $\hbar\in [0,1]$.  We similarly define a representation of the C*-tensor product $\pi^{\mathcal{E}_\hbar}_{\mathfrak{A}_\hbar\otimes\mathfrak{B}_\hbar}: \mathfrak{A}_\hbar\otimes\mathfrak{B}_\hbar\to \mathcal{B}(\mathcal{E}_\hbar)$ given by the continuous linear extension of
\begin{align}
\pi^{\mathcal{E}_\hbar}_{\mathfrak{A}_\hbar\otimes\mathfrak{B}_\hbar} = \pi^{\mathcal{E}_\hbar}_{\mathfrak{A_\hbar}}\otimes \pi^{\mathcal{E}_\hbar}_{\mathfrak{B}_\hbar}
\end{align}
at each $\hbar\in [0,1]$.\\

\begin{definition}
    Suppose $\mathfrak{A}\rightarrowtail\mathcal{E}\leftarrowtail\mathfrak{B}$ is a strongly non-degenerate Hilbert bimodule at each $\hbar\in [0,1]$. 
    \begin{itemize}
    \item An element $c\in \mathfrak{A}\otimes\mathfrak{B}$ is in the \emph{asymptotic center} relative to $\mathcal{E}$ if for all $c'\in\mathfrak{A}\otimes\mathfrak{B}$,
    \begin{align}
    \lim_{\hbar\to 0}\left\Vert \frac{i}{\hbar}\Big[\pi^{\mathcal{E}_\hbar}_{\mathfrak{A}_\hbar\otimes\mathfrak{B}_\hbar}\big(\Phi_\hbar(c)\big),\pi^{\mathcal{E}_\hbar}_{\mathfrak{A}_\hbar\otimes\mathfrak{B}_\hbar}\big(\Phi_\hbar(c')\big)\Big] \right \Vert_\hbar = 0,
    \end{align}
    where the norm on operators on $\mathcal{E}_\hbar$ is the operator norm on $\mathcal{B}(\mathcal{E}_\hbar)$, taken relative to the $\mathfrak{B}_\hbar$-valued inner product on $\mathcal{E}_\hbar$.
    \item The Hilbert bimodule $\mathcal{E}$ is \emph{asymptotically irreducible} if whenever $c\in \mathfrak{A}\otimes\mathfrak{B}$ is in the asymptotic center relative to $\mathcal{E}$, it follows that $\pi^{\mathcal{E}_0}_{\mathfrak{A}_0\otimes\mathfrak{B}_0}(\Phi_0(c))$ is a scalar multiple of the identity.
    \end{itemize}
\end{definition}

Now we prove our general result that under suitable conditions, the classical limit of a Hilbert bimodule is a symplectic dual pair.  Recall that we assume the situation where we begin with two uniformly continuous bundles of C*-algebras generated by strict deformation quantizations.  We take as given a strongly non-degenerate Hilbert bimodule between the C*-algebras of uniformly continuous sections.  We further assume that the commutative algebra $\mathcal{P}_S$ defined in Eq. (\ref{eq:smooth}) from the Hilbert classical limit satisfies the known conditions provided by Nestruev of being smooth \citep[][p. 37]{Ne03}, geometric \citep[][p. 23]{Ne03}, and complete \citep[][p. 31]{Ne03}, which is equivalent to $S$ being a smooth manifold  \citep[][p. 77, Thm. 7.2 and p. 79, Thm. 7.7]{Ne03}.  What remains is to establish that the Poisson structure on $S$ comes from a symplectic form.  We now show that when such a strongly non-degenerate Hilbert bimodule is asymptotically irreducible, the induced Poisson bracket comes from a symplectic form, which implies $S$ forms the middle space of a symplectic dual pair.

We note that it would be interesting to search for further conditions on the Hilbert bimodules at values $\hbar>0$ that might imply $S$ is a smooth manifold.  For example, one might hope to find analogues to Nestruev's conditions of being  smooth \citep[][p. 37]{Ne03}, geometric \citep[][p. 23]{Ne03}, and complete \citep[][p. 31]{Ne03} that apply for non-commutative algebras at $\hbar>0$ yet guarantee these conditions are satisfied in the classical $\hbar\to 0$ limit.  We lack the space here to undertake this investigation of induced smooth structure from properties at $\hbar>0$.  We aim here only to prove that if $S$ has an appropriate smooth manifold structure induced from the algebra $\mathcal{P}_S$, then asymptotic irreducibility guarantees this manifold is canonically a symplectic manifold.\\

\begin{theorem}
\label{thm:symplim}
Suppose $\mathfrak{A}\rightarrowtail\mathcal{E}\leftarrowtail\mathfrak{B}$ is a strongly non-degenerate Hilbert bimodule at each $\hbar\in[0,1]$.  The Hilbert classical limit of $\mathcal{E}$ is the Hilbert bimodule denoted by $\mathfrak{A}_0\rightarrowtail\mathcal{E}_0\leftarrowtail\mathfrak{B}_0$.  Suppose further that:
\begin{enumerate}[(i)]
\item $\mathcal{P}_S$ as defined in Eq. (\ref{eq:smooth}) is a smooth, geometric, and complete algebra; and
\item the Hilbert bimodule $\mathcal{E}$ is asymptotically irreducible.
\end{enumerate}
Then the structure
\begin{align}
M\leftarrow S\rightarrow N
\end{align}
defined in Eqs. (\ref{eq:gelfand})-(\ref{eq:proj}) is a symplectic dual pair.  In this case, we refer to $S$ as \emph{the (symplectic) classical limit} of $\mathcal{E}$.
\end{theorem}

\begin{proof}
\begin{enumerate}
    \item First, we establish that $S$ is symplectic with respect to the Poisson bracket defined in Eqs. (\ref{eq:poiss})-(\ref{eq:prodpoiss}).
    Suppose $f = f_M\otimes f_N \in \mathcal{P}_M\otimes \mathcal{P}_N$ is such that
\begin{align}
    \{f\circ \hat{\rho}, g\}_S = 0
\end{align}
for all $g\in \mathcal{P}_S$.
Define the map $\tilde{\mathcal{Q}}_\hbar: \mathfrak{A}_0\otimes\mathfrak{B}_0\to\mathfrak{A}_\hbar\otimes\mathfrak{B}_\hbar$ by
\begin{align}
    \tilde{\mathcal{Q}}_\hbar = \mathcal{Q}_\hbar\otimes\mathcal{Q}'_\hbar
\end{align}
Then for every $g = g_M\otimes g_N\in \mathcal{P}_M\otimes\mathcal{P}_N$, we know that
\begin{align}
    \lim_{\hbar\to 0} \left\Vert \frac{i}{\hbar}\Big[\pi_{\mathfrak{A}_\hbar\otimes\mathfrak{B}_\hbar}^{\mathcal{E}_\hbar}(\tilde{\mathcal{Q}}_\hbar(f)),\pi_{\mathfrak{A}_\hbar\otimes\mathfrak{B}_\hbar}^{\mathcal{E}_\hbar}(\tilde{\mathcal{Q}}_\hbar(g))\Big]\right\Vert_\hbar = 0.
\end{align}
Hence, since $\tilde{\mathcal{Q}}_\hbar[\mathcal{P}_M\otimes\mathcal{P}_N]$ is dense in $\mathfrak{A}_\hbar\otimes\mathfrak{B}_\hbar$, we have that for all $c\in \mathfrak{A}_\hbar\otimes\mathfrak{B}_\hbar$,
\begin{align}
    \lim_{\hbar\to 0} \left\Vert \frac{i}{\hbar}\Big[\pi_{\mathfrak{A}_\hbar\otimes\mathfrak{B}_\hbar}^{\mathcal{E}_\hbar}(\tilde{\mathcal{Q}}_\hbar(f)),\pi_{\mathfrak{A}_\hbar\otimes\mathfrak{B}_\hbar}^{\mathcal{E}_\hbar}(c)\Big]\right\Vert_\hbar = 0.
\end{align}
Since $\mathcal{E}$ is asymptotically irreducible, it follows that $\pi_{\mathfrak{A}_0\otimes\mathfrak{B}_0}^{\mathcal{E}_0}(f) = f\circ\hat{\rho} = \alpha I$ for $\alpha\in\mathbb{C}$.

Now, we reformulate what we have just shown in terms of the Poisson bivector $B$ defined by
\begin{align}
    B(dh_1\wedge dh_2) = \{h_1,h_2\}_S
\end{align}
for all $h_1,h_2\in \mathcal{P}_S$.  We have established that for any $f\in \mathcal{P}_M\otimes \mathcal{P}_N$, whenever
\begin{align}
    B(d(f\circ\hat{\rho})\wedge \eta) = 0
\end{align}
for all covector fields $\eta$ on $S$, it follows that $f\circ\hat{\rho}$ is a constant scalar field on $S$, and hence, $d(f\circ \hat{\rho}) = 0$.  Since this equation is local, we have that at any point $s\in S$, if
\begin{align}
    B_{|s}(d(f\circ\hat{\rho})_{|s}\wedge \eta_{|s}) = 0
\end{align}
for all covectors $\eta_{|s}$ at $s$, then $d(f\circ\hat{\rho})_{|s} = 0$.

Similarly, if $\tilde{f}\in \overline{\mathcal{P}_M}\otimes\overline{\mathcal{P}_N}$, then at each point $s\in S$, there is some neighborhood $U$ of $s$ and some $f\in \mathcal{P}_M\otimes\mathcal{P}_N$ such that $(\tilde{f}\circ \hat{\rho})_{|U} = (f\circ \hat{\rho})_{|U}$.  Hence, $d(f\circ \hat{\rho})_{|s} = d(\tilde{f}\circ \hat{\rho})_{|s}$.  It follows that if
\begin{align}
    B_{|s}(d(\tilde{f}\circ \hat{\rho})_{|s}\wedge \eta_{|s}) = B_{|s}(d(f\circ \hat{\rho})_{|s}\wedge \eta_{|s}) = 0
\end{align}
for all covectors $\eta_{|s}$ at $s\in S$, then $d(\tilde{f}\circ \hat{\rho})_{|s} = d(f\circ \hat{\rho})_{|s} = 0.$

Finally, we must consider the case of $f\in \overline{\overline{\mathcal{P}_M}\otimes\overline{\mathcal{P}_N}}$.  Then $f = g(f_1,...,f_k)$ for some $f_1,...,f_k\in \overline{\mathcal{P}_M}\otimes\overline{\mathcal{P}_N}$, and so for any $f'\in \mathcal{P}_S$, we have
\begin{align}
    B_{|s}(d(f\circ\hat{\rho})_{|s}\wedge df'_{|s}) = \{f\circ\hat{\rho},f'\}_S(s) &= \sum_{i=1}^k \frac{\partial g}{\partial x_i}_{|(f_1(s),...,f_k(s))} \{f_i,f'\}_{M\times N}(s)\nonumber\\
    &= \sum_{i=1}^k \frac{\partial g}{\partial x_i}_{|(f_1(s),...,f_k(s))} B_{|s}(d(f_i\circ \hat{\rho})_{|s} \wedge df'_{|s})
\end{align}
at each $s\in S$.  Since $\frac{\partial g}{\partial x_i}$ is independent of $f'$, it follows that whenever
\begin{align}
    B_{|s}(d(f\circ\hat{\rho})_{|s}\wedge \eta_{|s}) = 0
\end{align}
for all covectors $\eta_{|s}$ at $s\in S$, we have
\begin{align}
    B_{|s}(d(f_i\circ\hat{\rho})_{|s}\wedge \eta_{|s}) = 0
\end{align}
for all $i=1,...,k$, which implies $d(f_i\circ \hat{\rho})_{|s} = 0$ for all $i=1,...,k$ and thus, $d(f\circ\hat{\rho})_{|s} = 0.$

We conclude that $B_{|s}$ is non-degenerate for each $s\in S$, and hence $S$ is symplectic.

    \item Second, we establish the compatibility conditon for smooth functions on $M$ and $N$ by the pull-back to $S$.

    Consider $a\in\mathcal{P}_M$ and $b\in\mathcal{P}_N$.  We have 
    \begin{align}
        \{J_M^*(a),J_N^*(b)\}_S = \{(a\otimes I_N),(I_M\otimes b)\}_{M\times N} = \{a,I_M\}_M\otimes b + a\otimes\{I_N,b\} = 0,
    \end{align}
    where $I_M$ and $I_N$ are the constant identity functions on $M$ and $N$, respectively.
    
    Since the Poisson bracket is local, it follows that for any $a\in\overline{\mathcal{P}}_M$ and $b\in\overline{\mathcal{P}}_N$,
    \begin{align}
    \{J_M^*(a),J_N^*(b)\}_S = 0.
    \end{align}
\end{enumerate}
\end{proof}

\begin{cor}
\label{cor:reglim}
Suppose that $M\leftarrow S\rightarrow N$ is the symplectic classical limit of the Hilbert bimodule $\mathfrak{A}\rightarrowtail\mathcal{E}\leftarrowtail\mathfrak{B}$.  Suppose further that:
\begin{enumerate}[(i)]
\item the maps $J_M$ and $J_N$ defined in Eqs. (\ref{eq:proj}) are complete relative to the Poisson bracket; and
\item the map $J_N$ is a surjective submersion.
\end{enumerate}
Then $S$ is a weakly regular symplectic dual pair.\medskip
\end{cor}

Thus, we have shown that if one has a Hilbert bimodule $\mathfrak{A}_{0}\rightarrowtail\mathcal{E}_{0}\leftarrowtail\mathfrak{B}_{0}$ between the fibers at $\hbar=0$ that is the classical limit of a Hilbert bimodule between bundles generated by strict deformation quantizations, then under the conditions of Thm. \ref{thm:symplim}, one can naturally define a symplectic dual pair
\begin{align}
    M\leftarrow S \rightarrow N
\end{align}
that is the symplectic classical limit of the Hilbert bimodule.  It may be helpful to illustrate the symplectic classical limit with our running example.\medskip

\begin{example}
    Suppose $\beta_0: \mathfrak{A}_0\to\mathfrak{B}_0$ is a non-degenerate *-homomorphism satisfying $\beta(\mathcal{P}_M)\subseteq \mathcal{P}_N$ and preserving the Poisson bracket when restricted to a map from $\mathcal{P}_M$ to $\mathcal{P}_N$.  (This map $\beta_0$ may arise as the classical limit of what \citet[][p. 18, Def. 5.8]{StFe21a} call a \emph{smooth, second-order morphism of post-quantization bundles}.)  It follows that the dual map $\hat{\beta}_0: N\to M$ defined by $\beta_0(f)(n) = f(\hat{\beta}_0(n))$ for each $f\in \mathcal{P}_M$ and $n\in N$ is a smooth Poisson map between the Poisson manifolds $N$ and $M$.  Suppose further that $N$ is a symplectic manifold and $\hat{\beta}_0$ is an injective immersion.  Then the construction in Ex. \ref{ex:pmorphdualpair} yields from $\hat{\beta}_0$ an associated symplectic dual pair
    \begin{align}
    \label{eq:classdualpair}
        M\leftarrow \hat{\beta}_0(N)\times N\rightarrow N.
    \end{align}

    On the other hand, the construction in Ex. \ref{ex:*hombimod} yields from $\beta_0$ an associated Hilbert bimodule
    \begin{align}
    \label{eq:classbimod'}\mathfrak{A}_0\rightarrowtail\mathfrak{B}_0\leftarrowtail\mathfrak{B}_0
    \end{align}
    with $\pi_{\mathfrak{A}_0}^{\mathfrak{B}_0}(a)(b) = \beta_0(a)b$ and $\pi_{\mathfrak{B}_0}^{\mathfrak{B}_0}(b')b = bb'$ for all $a\in\mathfrak{A}_0$ and $b,b'\in\mathfrak{B}_0$.  We have the canonical *-isomorphism $\mathfrak{A}_0\otimes\mathfrak{B}_0\cong C_0(M\times N)$.  Furthermore, one can check that under this canonical *-isomorphism, we have
    \begin{align}
    \ker(\pi_{\mathfrak{A}_0}^{\mathfrak{B}_0}\otimes \pi_{\mathfrak{B}_0}^{\mathfrak{B}_0}) \cong \big\{ f\in C_0(M \times N)\ |\ f(p) = 0\text{ for all }p\in \hat{\beta}_0(N)\times N\big\},
    \end{align}
    which implies
    \begin{align}(\mathfrak{A}_0\otimes\mathfrak{B}_0)/\ker(\pi_{\mathfrak{A}_0}^{\mathfrak{B}_0}\otimes \pi_{\mathfrak{B}_0}^{\mathfrak{B}_0})\cong C_0\big(\hat{\beta}_0(N)\times N)\big).
    \end{align}
    It now follows according to Eqs. (\ref{eq:dualpairalgebra})-(\ref{eq:dualpairspace}), as in Thm. \ref{thm:symplim}, that the Hilbert bimodule in Eq. (\ref{eq:classbimod'}) is associated with the symplectic dual pair
    \begin{align}
        M\leftarrow \hat{\beta}_0(N)\times N\rightarrow N,
    \end{align}
    which is equivalent to the symplectic dual pair in Eq. (\ref{eq:classdualpair}).
    
    Hence, if a Hilbert bimodule is determined by a *-homomorphism whose classical limit corresponds to a smooth, injective Poisson immersion between symplectic manifolds, then the symplectic classical limit of that Hilbert bimodule agrees with the symplectic dual pair determined by that smooth, Poisson map.  Thus, the symplectic classical limit of Hilbert bimodules we have defined is an extension of the classical limit of *-homomorphisms from Steeger and Feintzeig \citep[][p. 19, Prop. 5.5]{StFe21a}.
\end{example}

\subsection{Functoriality of the Dual Symplectic Space}

Next, we must establish that this correspondence is functorial in the sense that it respects composition of symplectic dual pairs by the tensor product.  Again, we suppose as in Thm. \ref{thm:quotient} that $\mathfrak{A}$ and $\mathfrak{B}$ are the algebras of uniformly continuous sections of uniformly continuous bundles of C*-algebras over $[0,1]$, and we suppose $\mathfrak{A}\rightarrowtail\mathcal{E}\leftarrowtail\mathfrak{B}$ is a strongly non-degenerate Hilbert bimodule for $\alpha$ at $0\in [0,1]$.  We suppose further that the conditions of Thm. \ref{thm:symplim} are satisfied so that $\mathcal{E}$ has a symplectic classical limit $M\leftarrow S\rightarrow N$.

We will need to consider a further symplectic dual pair, so we now use the notation $S_1 = S$ for the middle space of the symplectic dual pair between $M$ and $N$ with projection maps $\overset{1}{J}_M = J_M$ and $\overset{1}{J}_N = J_N$.  So we will from now on denote this dual pair by
\begin{align}
    M\leftarrow S_1\rightarrow N.
\end{align}
We will next consider a further symplectic dual pair that it can be composed with.

Suppose now that $P$ is a further Poisson manifold with a strict deformation quantization $(\mathfrak{C}_\hbar,\mathcal{Q}_\hbar'')_{\hbar\in (0,1]}$ of $\mathcal{P}_P = C_c^\infty(P)$, which generates a uniformly continuous bundle of C*-algebras $((\mathfrak{C}_\hbar,\zeta_\hbar)_{\hbar\in(0,1]},\mathfrak{C})$.  Suppose further that we have a Hilbert bimodule
\begin{align}
    \mathfrak{B}\rightarrowtail \mathcal{F}\leftarrowtail\mathfrak{C}
\end{align}
that satisfies the assumptions of Thm. \ref{thm:symplim}, so that it has a Hilbert classical limit \emph{and} a corresponding symplectic classical limit.  We denote the Hilbert classical limit of $\mathcal{F}$ by
\begin{align} \mathfrak{B}_0\rightarrowtail\mathcal{F}_0\leftarrowtail\mathfrak{C}_0,
\end{align}
and we denote the symplectic classical limit of $\mathcal{F}$ by
\begin{align}
    N\leftarrow S_2\rightarrow P
\end{align}
with projection maps $\overset{2}{J}_N$ and $\overset{2}{J}_P$.
We shall further suppose that the assumptions of Cor. \ref{cor:reglim} are satisfied by $\mathcal{E}$ and $\mathcal{F}$, so that $S_1$ and $S_2$ are each weakly regular.  This ensures that their composition $S_1\circledcirc S_2$ is well-defined.

We consider the symplectic dual pair
\begin{align}
    M\leftarrow \mathcal{P}(\mathfrak{A}_0\circledast_{\mathcal{E}\otimes_{\mathfrak{B}}\mathcal{F}}\mathfrak{C}_0)\rightarrow P,
\end{align}
which is obtained by first taking the tensor product of the Hilbert bimodules and afterwards taking the symplectic classical limit.  We must compare this to the symplectic dual pair
\begin{align}
    M\leftarrow S_1\circledcirc_N S_2\rightarrow P,
\end{align}
which is obtained by first taking the symplectic classical limit of the bimodules and afterward taking the tensor product.  We will now show that these dual pairs are isomorphic.  We will need a number of preliminary lemmas.

First, we establish that on the algebraic side, the construction of the tensor product Hilbert bimodule arises essentially in three steps: taking the tensor product algebra, quotienting by an ideal, and then considering a subalgebra of the resulting quotient algebra.\medskip

\begin{lemma}
    \label{lem:tensprod}
    Let $\dot{\mathfrak{B}}_0 = \mathfrak{B}_0\oplus \mathbb{C} I_{\mathfrak{B}_0}$ denote the unitization of $\mathfrak{B}_0$.  Each representation of $\mathfrak{B}_0$ has a unique linear extension to the unitization $\dot{\mathfrak{B}}_0$, which we employ without change of notation. 
 Denote 
    \begin{align}
        \mathfrak{D}_0 = (\mathfrak{A}_0\circledast_{\mathcal{E}}\dot{\mathfrak{B}}_0)\otimes (\dot{\mathfrak{B}}_0\circledast_{\mathcal{F}}\mathfrak{C}_0).
    \end{align}
    Define $L_{\mathfrak{B}_0} = \overline{\text{\emph{span}} \big\{(a\otimes b_1)\otimes (b_2\otimes c)\in \mathfrak{D}_0\ |\ b_1\otimes b_2 \in \ker (\pi^{\mathcal{E}_0}_{\mathfrak{B}_0}\otimes \pi^{\mathcal{F}_0}_{\mathfrak{B}_0})\big\}}$ and
    \begin{align}
    \mathfrak{D}_0^\mathfrak{B} = \overline{\text{\emph{span}} \Big\{(a\otimes b_1) \otimes (b_2\otimes c) + L_{\mathfrak{B}_0} \in \mathfrak{D}_0/L_{\mathfrak{B}_0}\ |\ (b_1\otimes b_2 - I_{\mathfrak{B}_0}\otimes I_{\mathfrak{B}_0})\in \ker(\pi^{\mathcal{E}_0}_{\mathfrak{B}_0}\otimes\pi_{\mathfrak{B}_0}^{\mathcal{F}_0})\Big\}}.
    \end{align}
    Then
\begin{align}
\mathfrak{A}_0\circledast_{\mathcal{E}\otimes_{\mathfrak{B}}\mathcal{F}}\mathfrak{C}_0 \cong \mathfrak{D}_0^\mathfrak{B}.
\end{align}
\end{lemma}

\begin{proof}
    Let $K = \ker(\pi^{\mathcal{E}_0\otimes_{\mathfrak{B}_0}\mathcal{F}_0}_{\mathfrak{A}_0}\otimes \pi^{\mathcal{E}_0\otimes_{\mathfrak{B}_0}\mathcal{F}_0}_{\mathfrak{C}_0})$.  Define a map $\tau_1: \mathfrak{A}_0\circledast_{\mathcal{E}\otimes_{\mathfrak{B}}\mathcal{F}}\mathfrak{C}_0\to \mathfrak{D}_0^\mathfrak{B}$ by the continuous linear extension of the assignment
    \begin{align}
    \label{eq:t1}
        \tau_1(a\otimes c + K) = (a\otimes I_{\mathfrak{B}_0})\otimes (I_{\mathfrak{B}_0}\otimes c) + L_{\mathfrak{B}_0}.
    \end{align}
for all $a\in\mathfrak{A}_0$ and $c\in\mathfrak{C}_0$. 
 First, we must show this map is well-defined.  Consider $a\otimes c\in K$.  Then clearly 
\begin{align}
(a\otimes I_{\mathfrak{B}_0})\otimes (I_{\mathfrak{B}_0}\otimes c) \in \ker\Big((\pi^{\mathcal{E}_0\otimes_{\mathfrak{B}_0}\mathcal{F}_0}_{\mathfrak{A}_0}\otimes \pi^{\mathcal{E}_0\otimes_{\mathfrak{B}_0}\mathcal{F}_0}_{\mathfrak{B}_0})\otimes (\pi^{\mathcal{E}_0\otimes_{\mathfrak{B}_0}\mathcal{F}_0}_{\mathfrak{B}_0}\otimes \pi^{\mathcal{E}_0\otimes_{\mathfrak{B}_0}\mathcal{F}_0}_{\mathfrak{C}_0})\Big).
\end{align}
Hence, $\tau_1(a\otimes c) \in L_{\mathfrak{B}_0}$.

Notice further that $\tau_1$ is clearly a *-homomorphism.

Second, we show that $\tau_1$ is injective.  Suppose that $\tau_1(\sum_i a_i\otimes c_i) = 0$.  Then, since $I_{\mathfrak{B}_0}\otimes I_{\mathfrak{B}_0}\notin \ker(\pi_{\mathfrak{B}_0}^{\mathcal{E}_0}\otimes \pi_{\mathfrak{B}_0}^{\mathcal{F}_0})$, we must have
\begin{align}
\sum_i(a_i\otimes I_{\mathfrak{B}_0})\otimes (I_{\mathfrak{B}_0}\otimes c_i) \in \ker\Big((\pi^{\mathcal{E}_0\otimes_{\mathfrak{B}_0}\mathcal{F}_0}_{\mathfrak{A}_0}\otimes \pi^{\mathcal{E}_0\otimes_{\mathfrak{B}_0}\mathcal{F}_0}_{\mathfrak{B}_0})\otimes (\pi^{\mathcal{E}_0\otimes_{\mathfrak{B}_0}\mathcal{F}_0}_{\mathfrak{B}_0}\otimes \pi^{\mathcal{E}_0\otimes_{\mathfrak{B}_0}\mathcal{F}_0}_{\mathfrak{C}_0})\Big).
\end{align}
It follows from this and the definition of $\mathcal{E}\otimes_{\mathfrak{B}_0}\mathcal{F}_0$ that $\sum_i a_i\otimes c_i\in K$.

Finally, we show that $\tau_1$ is surjective.  Consider an arbitrary
$(a\otimes b_1)\otimes (b_2\otimes c)\in \mathfrak{D}_0^\mathfrak{B}$.  Then it follows that
\begin{align}
    (a\otimes b_1)\otimes (b_2\otimes c) - (a\otimes I_{\mathfrak{B}_0})\otimes (I_{\mathfrak{B}_0}\otimes c) \in L_{\mathfrak{B}}.
\end{align}
Hence $\tau_1(a\otimes c + K) = (a\otimes b_1)\otimes (b_2\otimes c)$.

We conclude that $\tau_1$ is a *-isomorphism.
\end{proof}

Next, we establish that the steps outlined in the previous lemma for the algebraic construction of the tensor product Hilbert bimodule have analog or dual steps in the geometric construction of a tensor product dual pair.  In particular, we know first of all that taking the tensor product algebra corresponds to taking a Cartesian product $S_1\times S_2$.  The following lemma shows that the next step of quotienting by an ideal corresponds to looking at the subset $S_1\times_N S_2$ of the Cartesian product space $S_1\times S_2$.\medskip

\begin{lemma}
\label{lem:sympprod}
    Let $L_N  = \big\{f\in C_0(S_1\times S_2)\ |\ f(s_1,s_2) = 0\text{ whenever }\overset{1}{J}_N(s_1) = \overset{2}{J}_N(s_2)\big\}$.  Then
    \begin{align}
        C_0(S_1\times_N S_2) \cong C_0(S_1\times S_2)/L_N.
    \end{align}
\end{lemma}

\begin{proof}
    Let $\iota_{S_1\times_N S_2}$ be the identical embedding of $S_1\times_N S_2$ in $S_1\times S_2$.  Define a map $\tau_2: C_0(S_1\times S_2)/L_N\to C_0(S_1\times_N S_2)$ by
    \begin{align}
    \label{eq:t2}
    \tau_2(f + L_N) = f\circ \iota_{S_1\times_N S_2}
    \end{align}
    for all $f\in C_0(S_1\times_N S_2)$.  Note that $\tau_2$ is well-defined.  If $f + L_N = g + L_N$ for $f,g\in C_0(S_1\times S_2)$, then $(f-g)(s) = 0$ for all $s\in S_1\times_N S_2$, and hence $f\circ \iota_{S_1\times_N S_2} = g\circ \iota_{S_1\times_N S_2}$.  Moreover, clearly $\tau_2$ is a *-homomorphism.
    
    Similarly, $\tau_2$ is injective, for whenever $\tau_2(f + L_N) = f\circ \iota_{S_1\times_N S_2} = 0$, it follows that $f\in L_N$ so $f+L_N = 0 + L_N$.
    
    To see that $\tau_2$ is surjective, consider an arbitrary $f\in C_0(S_1\times_N S_2)$.  Since $S_1\times_N S_2$ is closed in $S_1\times S_2$, the Tietze extension theorem implies that there is a function $\tilde{f}\in C_0(S_1\times S_2)$ such that $\tilde{f}\circ \iota_{S_1\times_N S_2} = f$, and hence $\tau_2(\tilde{f} + L_N) = f$.

    We conclude that $\tau_2$ is a *-isomorphism.
\end{proof}

The following lemma shows that in the final step, considering a certain subalgebra of functions on $S_1\times_N S_2$ corresponds to taking a quotient of the space $S_1\times_N S_2$ by looking at the space of leaves $S_1\circledcirc S_2$ of the foliation by $\mathcal{N}_{S_1\times_N S_2}$.\medskip

\begin{lemma}
\label{lem:sympfol}
     \begin{align}
        C_0(S_1\circledcirc_N S_2)\cong \Big\{f\in C_0(S_1\times_N S_2)\ |\ \xi(f) = 0\text{ whenever } \xi\in\mathcal{N}_{S_1\times_N S_2}\Big\}.
     \end{align}
\end{lemma}

\begin{proof}
   Let $q_{\mathcal{N}_{S_1\times_N S_2}}: S_1\times_N S_2\to S_1\circledcirc_N S_2$ denote the map that sends each point $(s_1,s_2) \in S_1\times_N S_2$ to the leaf $[(s_1,s_2)]$ of the foliation by $\mathcal{N}_{S_1\times_N S_2}$ that it belongs to.  Define a map $\tau_3: C_0(S_1\circledcirc_N S_2)\to C_0(S_1\times_N S_2)$ by
    \begin{align}
    \label{eq:t3}
        \tau_3(f) = f\circ q_{\mathcal{N}_{S_1\times_N S_2}}
    \end{align}
for all $f\in C_0(S_1\circledcirc_N S_2)$.  Clearly, $\tau_3$ is a *-homomorphism.  Moreover, if $f\in C_0(S_1\circledcirc_N S_2)$, then $\tau_3(f)$ is constant on leaves of the foliation by $\mathcal{N}_{S_1\times_N S_2}$, so that $\xi(\tau_3(f)) = 0$ for all $\xi\in \mathcal{N}_{S_1\times_N S_2}$, which implies $\tau_3$ has the appropriate range.

The map $\tau_3$ is injective because $\tau_3(f) = 0$ on $S_1\times _N S_2$ implies that $f$ vanishes on each leaf, which means that $f=0$.

Finally, $\tau_3$ is surjective because whenever $f\in C_0(S_1\times_N S_2)$ satisfies $\xi(f) = 0$ for all $\xi \in \mathcal{N}_{S_1\times_N S_2}$, it follows that $f$ is constant on each leaf of the foliation by $\mathcal{N}_{S_1\times_N S_2}$.  Hence, the function $\tilde{f}\in C_0(S_1\circledcirc_N S_2)$ given by $\tilde{f}([(s_1,s_2)]) = f(s_1,s_2)$ is well-defined and satisfies $\tau_3(\tilde{f}) = f$.

We conclude that $\tau_3$ is a *-isomorphism.
\end{proof}

We will also need the following result characterizing the leaves of $S_1\times_N S_2$.\medskip

\begin{lemma}
\label{lem:leaves}
    Suppose $\gamma$ is a smooth curve in $S_1\times_N S_2$.  The tangent vectors to $\gamma$ lie in $\mathcal{N}_{S_1\times_N S_2}$ if and only if for any two points $\gamma(t_1) = (s_1,s_2)$ and $\gamma(t_2) = (s_1',s_2')$, we have
    \begin{align}
        \overset{1}{J}_M(s_1) = \overset{1}{J}_M(s_1') && \overset{2}{J}_P(s_2) = \overset{2}{J}_P(s_2').
    \end{align}
    Hence, the leaves of the foliation determined by $\mathcal{N}_{S_1\times_N S_2}$ have the form
    \begin{align}
        O_{m,p} = \Big\{(s_1,s_2)\in S_1\times_N S_2\ |\ \overset{1}{J}_M(s_1) = m,\ \overset{2}{J}_P(s_2) = p\Big\}
    \end{align}
    for some fixed $m\in M$ and $p\in P$.
\end{lemma}

\begin{proof}
We begin with some preliminaries.  Denote the projections $pr_{S_1}: S_1\times_N S_2\to S_1$ and $pr_{S_2}: S_1\times_N S_2\to S_2$.  Suppose $\gamma$ is a smooth curve in $S_1\times_N S_2$ and denote the tangent vector at $\gamma(t)$ by $\vec{\gamma}_{| t}$.

The definition of $S_1\times_N S_2$ yields
\begin{align}
    (\overset{1}{J}_N\circ pr_{S_1})(s) = (\overset{2}{J}_N\circ pr_{S_2})(s)
\end{align}
for all $s\in S_1\times_N S_2$, so that for any vector $\xi\in T(S_1\times_N S_2)$ at $\gamma(t)$ and any $f\in C^\infty(N)$, we have
\begin{align}
    (\overset{1}{J}_N\circ pr_{S_1})_*(\xi)(f) = \xi(f\circ \overset{1}{J}_N\circ pr_{S_1}) = \xi(f\circ \overset{2}{J}_N\circ pr_{S_2}) = (\overset{2}{J}_N\circ pr_{S_2})_*(\xi)(f).
\end{align}

Hence, denoting by $\omega_1$ and $\omega_2$ the symplectic forms of $S_1$ and $S_2$ at $s_1 = pr_{S_1}(\gamma(t))$ and $s_2 = pr_{S_2}(\gamma(t))$, respectively, we have from the definition of the product symplectic form that
\begin{align}
\label{eq:o12}
    (\omega_1\oplus \omega_2)(\vec{\gamma}_{|t},\xi) &= pr_{S_1}^*(\omega_1)(\vec{\gamma}_{|t},\xi) + pr_{S_2}^*(\omega_2)(\vec{\gamma}_{|t},\xi)\nonumber\\
    &= \omega_1\Big((pr_{S_1})_*\vec{\gamma}_{|t},(pr_{S_1})_*\xi\Big) + \omega_2\Big((pr_{S_2})_*\vec{\gamma}_{|t},(pr_{S_2})_*\xi\Big). 
\end{align}
We also know that $\omega_1$ is a restriction of the Poisson structure on $M\times N$ and $\omega_2$ is a restriction of the Poisson structure on $N\times P$.  We use the isomorphisms
\begin{align}
    T_{s_1}S_1 \cong (\overset{1}{J}_M)_*(T_{s_1}S_1)\oplus T_nN && T_{s_2}S_2 \cong (\overset{2}{J}_N)_*(T_{s_2}S_2)\oplus T_pP,
\end{align}
where $n = \overset{1}{J}_N(s_1)$ and $p = \overset{2}{J}_P(s_2)$.  Since $J_M$-fibers and $J_N$-fibers are symplectically orthogonal in $S_1$, it follows that
\begin{align}
\label{eq:o1}
\omega_1\Big((pr_{S_1})_*\vec{\gamma}_{|t},(pr_{S_1})_*\xi\Big) = \omega_1&\Big((\overset{1}{J}_M\circ pr_{S_1})_*\vec{\gamma}_{|t},(\overset{1}{J}_M\circ pr_{S_1})_*\xi\Big)\nonumber\\
& + \omega_1\Big((\overset{1}{J}_N\circ pr_{S_1})_*\vec{\gamma}_{|t},(\overset{1}{J}_N\circ pr_{S_1})_*\xi\Big).
\end{align}
Similarly, since $\overset{2}{J}_N$-fibers and $\overset{2}{J}_P$-fibers are symplectically orthogonal in $S_2$, it follows that
\begin{align}
\label{eq:o2}
\omega_2\Big((pr_{S_2})_*\vec{\gamma}_{|t},(pr_{S_2})_*\xi\Big) = \omega_2&\Big((\overset{2}{J}_N\circ pr_{S_2})_*\vec{\gamma}_{|t},(\overset{2}{J}_N\circ pr_{S_2})_*\xi\Big)\nonumber\\
& + \omega_2\Big((\overset{2}{J}_P\circ pr_{S_2})_*\vec{\gamma}_{|t},(\overset{2}{J}_P\circ pr_{S_2})_*\xi\Big).
\end{align}

Now, denote the Poisson bivectors associated with $\omega_1$ and $\omega_2$ by $\overset{1}{B}$ and $\overset{2}{B}$, respectively.  Since $\omega_1$ and $\omega_2$ are symplectic, we have associated isomorphisms
\begin{align}
    \overset{1}{B}_\#&: TS_1\to T^*S_1\nonumber\\
   \overset{2}{B}_\#&: TS_2\to T^*S_2.
\end{align}
  Denoting by $\overset{M}{B}$, $\overset{N}{B}$, and $\overset{P}{B}$ the Poisson bivectors on $M$, $N$, and $P$, respectively, we have
\begin{align}
    \overset{1}{B} &= \overset{M}{B}  - \overset{N}{B} \nonumber\\
    \overset{2}{B} &= \overset{N}{B} - \overset{P}{B}.
\end{align}
Hence, it follows from Eq. (\ref{eq:o1}) that
\begin{align}
\label{eq:o1'}
\omega_1(&(pr_{S_1})_*\vec{\gamma}_{|t},(pr_{S_1})_*\xi)\nonumber\\
&= \overset{1}{B}\Big(\overset{1}{B}_\#((\overset{1}{J}_M\circ pr_{S_1})_*\vec{\gamma}_{|t}),\overset{1}{B}_\#((\overset{1}{J}_M\circ pr_{S_1})_*\xi)\Big) + \overset{1}{B}\Big(\overset{1}{B}_\#((\overset{1}{J}_N\circ pr_{S_1})_*\vec{\gamma}_{|t}),\overset{1}{B}_\#((\overset{1}{J}_N\circ pr_{S_1})_*\xi)\Big)\nonumber\\
&= \overset{M}{B}\Big(\overset{1}{B}_\#((\overset{1}{J}_M\circ pr_{S_1})_*\vec{\gamma}_{|t}),\overset{1}{B}_\#((\overset{1}{J}_M\circ pr_{S_1})_*\xi)\Big) -\overset{N}{B}\Big(\overset{1}{B}_\#((\overset{1}{J}_N\circ pr_{S_1})_*\vec{\gamma}_{|t}),\overset{1}{B}_\#((\overset{1}{J}_N\circ pr_{S_1})_*\xi)\Big)
\end{align}
and it follows from Eq. (\ref{eq:o2}) that
\begin{align}
\label{eq:o2'}
    \omega_2(&(pr_{S_2})_*\vec{\gamma}_{|t},(pr_{S_2})_*\xi)\nonumber\\
    &= \overset{2}{B}\Big(\overset{2}{B}_\#((\overset{2}{J}_N\circ pr_{S_2})_*\vec{\gamma}_{|t}),\overset{2}{B}_\#((\overset{2}{J}_N\circ pr_{S_2})_*\xi)\Big) + \overset{2}{B}\Big(\overset{2}{B}_\#((\overset{2}{J}_P\circ pr_{S_2})_*\vec{\gamma}_{|t}),\overset{2}{B}_\#((\overset{2}{J}_P\circ pr_{S_2})_*\xi)\Big)\nonumber\\
    &= \overset{N}{B}\Big(\overset{2}{B}_\#((\overset{2}{J}_N\circ pr_{S_2})_*\vec{\gamma}_{|t}),\overset{2}{B}_\#((\overset{2}{J}_N\circ pr_{S_2})_*\xi)\Big) - \overset{P}{B}\Big(\overset{2}{B}_\#((\overset{2}{J}_P\circ pr_{S_2})_*\vec{\gamma}_{|t}),\overset{2}{B}_\#((\overset{2}{J}_P\circ pr_{S_2})_*\xi)\Big).
\end{align}
Therefore, it follows from Eqs. (\ref{eq:o12}) and (\ref{eq:o1'})-(\ref{eq:o2'}) that
\begin{align}
\label{eq:o12'}
(\omega_1\oplus \omega_2)(\vec{\gamma}_{|t},\xi) = \overset{M}{B}\Big(&\overset{1}{B}_\#((\overset{1}{J}_M\circ pr_{S_1})_*\vec{\gamma}_{|t}),\overset{1}{B}_\#((\overset{1}{J}_M\circ pr_{S_1})_*\xi)\Big) \nonumber\\
&- \overset{P}{B}\Big(\overset{2}{B}_\#((\overset{2}{J}_P\circ pr_{S_2})_*\vec{\gamma}_{|t}),\overset{2}{B}_\#((\overset{2}{J}_P\circ pr_{S_2})_*\xi)\Big).
\end{align}

Now, the constraint that for any two points $\gamma(t_1) = (s_1,s_2)$ and $\gamma(t_2) = (s_1',s_2')$,
\begin{align}
    \overset{1}{J}_M(s_1) = \overset{1}{J}_M(s_1') && \overset{2}{J}_P(s_2) = \overset{2}{J}_P(s_2')
\end{align}
holds if and only if for all $h\in C^\infty(M)$ and $g\in C^\infty(P)$, we have
\begin{align}
    (\overset{1}{J}_M\circ pr_{S_1})_*(\vec{\gamma}_{|t})(h) &= \vec{\gamma}_{|t}(h\circ \overset{1}{J}_M\circ pr_{S_1}) = 0\nonumber\\
    (\overset{2}{J}_P\circ pr_{S_2})_*(\vec{\gamma}_{|t})(g) &= \vec{\gamma}_{|t}(g\circ \overset{2}{J}_P\circ pr_{S_2}) = 0,
\end{align}
or equivalently, $(\overset{1}{J}_M\circ pr_{S_1})_*(\vec{\gamma}_{|t}) = 0$ and $(\overset{2}{J}_P\circ pr_{S_2})_*(\vec{\gamma}_{|t}) = 0$.  It follows from Eq. (\ref{eq:o12'}) that this constraint holds if and only if for all $\xi\in T(S_1\times_N S_2)$,
\begin{align}
    (\omega_1\oplus \omega_2)(\vec{\gamma}_{|t},\xi) = 0,
\end{align}
which by definition holds if and only if $\vec{\gamma}_{|t}\in \mathcal{N}_{S_1\times_N S_2}$.
\end{proof}

We can now use the understanding provided by Lemma \ref{lem:tensprod} of the construction of a tensor product Hilbert bimodule and the understanding provided by Lemmas \ref{lem:sympprod} and \ref{lem:sympfol} of the construction of the tensor product dual pair to prove that the classical limit of a Hilbert bimodule is functorial.  Essentially, we show that the three step procedures for constructing the tensor product Hilbert bimodule and the tensor product dual pair correspond with each other.\medskip

\begin{theorem}
\label{thm:funclim}
Suppose $S_1$ and $S_2$ are weakly regular symplectic dual pairs determined by the symplectic classical limit of Hilbert bimodules $\mathfrak{A}\rightarrowtail\mathcal{E}\leftarrowtail\mathfrak{B}$ and $\mathfrak{B}\rightarrowtail\mathcal{F}\leftarrowtail\mathfrak{C}$, respectively.  Then the symplectic dual pairs
\begin{align}
    M\leftarrow \mathcal{P}(\mathfrak{A}_0\circledast_{\mathcal{E}\otimes_{\mathfrak{B}}\mathcal{F}}\mathfrak{C}_0)\rightarrow P,
\end{align}
and 
\begin{align}
    M\leftarrow S_1\circledcirc_N S_2\rightarrow P,
\end{align}
are isomorphic.
\end{theorem}

\begin{proof}
We will show that for the tensor product dual pair $S_1\circledcirc_N S_2$, we have an isomorphism
    \begin{align}
        C_0(S_1\circledcirc_N S_2) \cong \mathfrak{A}_0\circledast_{\mathcal{E}\otimes_{\mathfrak{B}}\mathcal{F}}\mathfrak{C}_0.
    \end{align}
Hence, Gelfand duality implies
\begin{align}
    S_1\circledcirc_N S_2\cong \mathcal{P}(\mathfrak{A}_0\circledast_{\mathcal{E}\otimes_{\mathfrak{B}}\mathcal{F}}\mathfrak{C}_0).
\end{align}

We once again denote $K= \ker\big(\pi_{\mathfrak{A}_0}^{\mathcal{E}_0\otimes_{\mathfrak{B}_0}\mathcal{F}_0} \otimes\pi_{\mathfrak{C}_0}^{\mathcal{E}_0\otimes_{\mathfrak{B}_0}\mathcal{F}_0}\big)$, as above.  Now we explicitly define an isomorphism $u: \mathfrak{A}_0\circledast_{\mathcal{E}\otimes_{\mathfrak{B}}\mathcal{F}}\mathfrak{C}_0\to C_0(S_1\circledcirc_N S_2)$ by
\begin{align}
    u(a\otimes c + K)\big([(s_1,s_2)]\big) = a\big(\overset{1}{J}_M(s_1)\big)\cdot c\big(\overset{2}{J}_P(s_2)\big)
\end{align}
for any $a\in\mathfrak{A}_0 = C_0(M)$ and $c\in\mathfrak{C}_0= C_0(P)$ at all points $s_1\in S_1$ and $s_2\in S_2$.  Here, $[(s_1,s_2)]$ is the leaf of the foliation by $\mathcal{N}_{S_1\times_N S_2}$ through the point $(s_1,s_2)\in S_1\times_N S_2$.

We must first show $u$ is well-defined. 
 Notice that it follows from Lemma \ref{lem:leaves} that the right hand side is indeed constant on symplectic leaves. 
 Also, consider $a\otimes c\in K$.  Then it follows from Eq. (\ref{eq:t1}) that $\tau_1(a\otimes c)\in L_{\mathfrak{B}}$.  It follows that 
 \begin{align*}    
 (a\otimes I_{\mathfrak{B}_0})(s_1)\cdot (I_{\mathfrak{B}_0}\otimes c)(s_2) = 0
 \end{align*}
 for all $s_1\in S_1$ and $s_2\in S_2$.  This implies $u(a\otimes c + K) = 0$.  This establishes that $u$ is well-defined, and clearly $u$ is a *-homomorphism.

 Next, we show $u$ is injective.  If $u(a\otimes c + K) = 0$, then it follows that $\tau_1(a\otimes c) \in L_{\mathfrak{B}}$, so the injectivity of $\tau_1$ implies $a\otimes c\in K$.

 Finally, we show that $u$ is surjective.  Consider an arbitrary $f\in C_0(S_1\circledcirc_N S_2)$.  Then, according to Eqs. (\ref{eq:t2}) and (\ref{eq:t3}), we have $\tau_3(f) = \tau_2(\tilde{f} + L_N)$ for some $\tilde{f}\in C_0(S_1\times S_2)$ since $\tau_2$ is surjective.  But we know $C_0(S_1\times S_2)\cong C_0(S_1)\otimes C_0(S_2)$, so there are sequences $a_n\in\mathfrak{A}_0$, $b_m,b_j\in\mathfrak{B}_0$ and $c_k\in\mathfrak{C}_0$ such that
 \begin{align}
     \tilde{f}(s_1,s_2) &= \lim_{n,m,j,k}\Big(\sum_{n,m} \pi_{\mathfrak{A}_0}^{\mathcal{E}_0}(a_n)\otimes \pi_{\mathfrak{B}_0}^{\mathcal{E}_0}(b_m)\Big)(s_1) \otimes \Big(\sum_{j,k} \pi_{\mathfrak{B}_0}^{\mathcal{F}_0}(b_j)\otimes \pi_{\mathfrak{C}_0}^{\mathcal{F}_0}(c_k)(s_2)\Big)\nonumber\\
     &= \lim_{n,m,j,k}\Big(\sum_{n,m} a_n\big(\overset{1}{J}_M(s_1)\big)\cdot b_m\big(\overset{1}{J}_N(s_1)\big)\Big)\otimes \Big(\sum_{j,k} b_j\big(\overset{2}{J}_N(s_2)\big)\cdot c_k\big(\overset{2}{J}_P(s_2)\big)\Big)\nonumber\\
     &= C \lim_{n,k}\sum_{n,k} a_n\big(\overset{1}{J}_M(s_1)\big)\cdot c_k\big(\overset{2}{J}_P(s_2)\big), 
 \end{align}
where $C = \sum_{m,j} b_m\big(\overset{1}{J}_N(s_1)\big)\cdot b_j\big(\overset{2}{J}_N(s_2)\big)$ is constant since $\tau_3(f)$ is constant on leaves.
Now, we have
\begin{align}
u\Big(C\lim_{n,k}\sum_{n,k}(a_n\otimes c_k) + K\Big)\big([(s_1,s_2)]\big) &= C\lim_{n,k}\sum_{n,k} a_n\big(\overset{1}{J}_M(s_1)\big)\cdot c_k\big(\overset{2}{J}_P(s_2)\big)\nonumber\\
&=\tilde{f}(s_1,s_2)\nonumber\\
&=\tau_2(\tilde{f} + L_N)(s_1,s_2)\nonumber\\
&=\tau_3(f)(s_1,s_2)\nonumber\\
&= f\big([(s_1,s_2)]\big).
\end{align}
This establishes that $u$ is surjective, and hence a *-isomorphism.

It is now straightforward to see that  by construction, the Gelfand dual $\hat{u}: S_1\circledcirc_N S_2\to \mathcal{P}\big(\mathfrak{A}_0\circledast_{\mathcal{E}\otimes_{\mathfrak{B}}\mathcal{F}}\mathfrak{C}_0\big)$ preserves the Poisson bracket, and hence the symplectic form, as well as the projection maps.  It follows that $\hat{u}$ realizes an isomorphism of symplectic dual pairs.
\end{proof}

This establishes the functoriality of the symplectic classical limit.

\section{Conclusion}
\label{sec:con}

In this paper, we have provided a two-step procedure for taking the classical limit of a Hilbert bimodule, understood as a morphism between C*-algebras representing quantum theories in the framework of strict deformation quantization.  One starts with a Hilbert bimodule between algebras of uniformly continuous sections of uniformly continuous bundles of C*-algebras, whenever the Hilbert bimodule satisfies the condition of continuous scaling.  In the first step of our procedure, we constructed a quotient Hilbert bimodule that serves as a morphism between the abelian C*-algebras that are the classical limit of the non-commutative algebras.  In the second step of the procedure, we constructed a symplectic space that is dual to that quotient Hilbert bimodule, which serves as a symplectic dual pair between the corresponding classical Poisson manifolds.  We showed that each of these steps is functorial.  In the first step, we showed that composition of Hilbert bimodules between the typically non-commutative C*-algebras at $\hbar>0$ is preserved as composition of Hilbert bimodules between commutative C*-algebras at $\hbar = 0$ through the classical limit.  In the second step, we showed that composition of Hilbert bimodules between commutative C*-algebras at $\hbar = 0$ is preserved as composition of symplectic dual pairs between classical Poisson manifolds whenever those dual pairs are weakly regular so that their composition is defined.  Putting this together, we have shown that the classical limit is functorial for Hilbert bimodules and symplectic dual pairs, when composition is understood via the appropriate tensor product in each case.

We must make a remark concerning the second step of our procedure in which one employs Gelfand duality to construct the middle space of a symplectic dual pair from a Hilbert bimodule.  The assumption we made to guarantee that this space is a symplectic manifold---in particular, the hypothesis of asymptotic irreducibility---is quite strong.  The condition of asymptotic irreducibility certainly holds in the simplest case when the resulting dual pair arises from a symplectomorphism as in Ex. \ref{ex:pmorphdualpair}.  In more general cases, the assumption of asymptotic irreducibility may fail and one may require a more general construction to obtain a symplectic manifold to serve as the middle space of a dual pair.  In such a case, we conjecture that one might search for a larger C*-algebra of operators on the Hilbert bimodule that retains the analogous property of asymptotic irreducibility.  We leave this as a topic for future work to study and perhaps generalize our construction to a larger class of Hilbert bimodules.

Among the Hilbert bimodules we would like to consider are the ones constructed by \citet[][p. 106, Lemma 1]{La01a} between Lie groupoid C*-algebras, as in Ex. \ref{ex:quantYMbimod}.  Indeed, \citet{La02a} shows that these arise as the quantization of certain symplectic dual pairs between the duals of Lie algebroids as in Ex. \ref{ex:classYMdualpair}.  Moreover, this quantization is functorial.  In future work, we hope to apply the tools presented here to analyze the classical limits of those Hilbert bimodules between Lie groupoid C*-algebras.  Indeed, we conjecture that our classical limit functor is almost inverse to the quantization functor defined by Landsman in the sense that the classical limit should reconstruct the symplectic dual pairs betweem duals of Lie algebroids that he quantizes.

\section*{Acknowledgments}
The authors thank Kade Cicchella and Michael Clancy for helpful discussion that led to this work.  Both authors were supported by the National Science Foundation under Grant nos. 1846560 and 2043089.

\section*{Data Availability Statement}

Data sharing not applicable to this article as no datasets were generated or analysed during the current study.

\appendix

\section{The Linking Algebra}
\label{app:link}

In this appendix, we establish some technical lemmas necessary for the proof of Thm. \ref{thm:quotient}.  To do so, we construct a linking algebra analogous to that discussed in Raeburn and Williams \citep[][p. 50]{RaWi98}. As in Thm. \ref{thm:quotient}, we suppose $\mathfrak{A}$ and $\mathfrak{B}$ are the total spaces of continuous bundles of C*-algebras over base spaces given by locally compact metric spaces $I$ and $J$, respectively.  We suppose that $\alpha: I\to J$ is a continuous proper map so that we have corresponding ideals $K_\hbar$ in $\mathfrak{A}$ and $K_{\alpha(\hbar)}$ in $\mathfrak{B}$ defined in Eqs. (\ref{eq:idealA})-(\ref{eq:idealB}).  With this in place, we suppose $\mathfrak{A}\rightarrowtail\mathcal{E}\leftarrowtail\mathfrak{B}$ is a strongly non-degenerate Hilbert bimodule for $\alpha$ at $\hbar\in I$.

The results we need involve a so-called linking algebra $\mathfrak{L}$, which we define as follows.  For $\varphi\in\mathcal{E}$, define the map $\flat_\varphi: \mathcal{E}\to\mathfrak{B}$ by
\begin{align}
\flat_\varphi(\eta) = \inner{\varphi}{\eta}_{\mathfrak{B}}
\end{align}
for all $\eta\in\mathcal{E}$.  Then for $R\in\mathcal{K}(\mathcal{E})$, $\varphi,\eta\in\mathcal{E}$, and $b\in\mathfrak{B}$, consider the matrix \begin{align}
L(R,\varphi,\eta,b) = \begin{pmatrix}
R & \eta\\
\flat_\varphi & b
\end{pmatrix}
\end{align}
understood as a linear operator on $ \mathcal{E}\oplus\mathfrak{B}$ acting by
\begin{align}
\begin{pmatrix}
R & \eta\\
\flat_\varphi & b
\end{pmatrix}
\begin{pmatrix}
\xi\\
b'
\end{pmatrix} = \begin{pmatrix}
    R\xi + \eta\cdot b'\\
    \inner{\varphi}{\xi}_{\mathfrak{B}} + bb'
\end{pmatrix}
\end{align}
for all $\xi\in\mathcal{E}$ and $b'\in\mathfrak{B}$.  Let $\mathfrak{L}$ denote the set of all operators of the form $L(R,\varphi,\eta,b)$ for some $R\in\mathcal{K}(\mathcal{E})$, $\varphi,\eta\in\mathcal{E}$, and $b\in\mathfrak{B}$.  Finally, notice that $\mathcal{E}\oplus\mathfrak{B}$ carries a $\mathfrak{B}$-valued inner product given by
\begin{align}
\inner{\varphi\oplus b}{\eta\oplus b'}_\mathfrak{B}  = \inner{\varphi}{\eta}_{\mathfrak{B}} + b^*b'
\end{align}
for all $\varphi,\eta\in\mathcal{E}$ and $b,b'\in\mathfrak{B}$.
Then we have the following result.\medskip

\begin{lemma}
\label{lemma:link}
$\mathfrak{L}$ is a C*-subalgebra of $\mathcal{L}(\mathcal{E}\oplus\mathfrak{B})$, and
\begin{align}
\label{eq:Lbound}
\max\big\{\norm{R},\norm{\varphi}_{\mathfrak{B}},\norm{\eta}_{\mathfrak{B}},\norm{b}\big\} \leq \norm{L(R,\varphi,\eta,b)}\leq \norm{R} + \norm{\varphi}_\mathfrak{B} + \norm{\eta}_{\mathfrak{B}} + \norm{b}.
\end{align}
\end{lemma}

\begin{proof}
First, note that each operator $L(R,\varphi,\eta,b)$ is indeed adjointable with adjoint $L(R^*,\eta,\varphi,b^*)$.  Clearly, $\mathfrak{L}$ is closed under addition and scalar multiplication.  Moreover, for another operator 
of the form $L(R',\varphi',\eta',b')$, we have
\begin{align}
L(R,\varphi,\eta,b)L(R',\varphi',\eta',b') = \begin{pmatrix}
R & \eta\\
\flat_\varphi & b
\end{pmatrix}
\begin{pmatrix}
R' & \eta'\\
\flat_{\varphi'} & b'
\end{pmatrix} = \begin{pmatrix}
RR' + \eta\flat_{\varphi'} & R\eta' + \eta\cdot b'\\
\flat_\varphi R' + b\flat_{\varphi'} & \inner{\varphi}{\eta}_\mathfrak{B} + bb'
\end{pmatrix}.
\end{align}
Now, $RR'\in\mathcal{K}(\mathcal{E})$ and $\eta\flat_{\varphi'}\in\mathcal{K}(\mathcal{E})$ so that $(RR' + \eta\flat_{\varphi'})\in\mathcal{K}(\mathcal{E})$.  Likewise, $(R\eta' + \eta\cdot b)\in \mathcal{E}$.  Moreover, $\flat_\varphi R' + b\flat_{\varphi'} = \flat_{R'^*\varphi + \varphi'\cdot b^*}$.  Finally, $(\inner{\varphi}{\eta}_{\mathfrak{B}} + bb')\in\mathfrak{B}$.  Hence, we have
\begin{align}
L(R,\varphi,\eta,b)L(R',\varphi',\eta',b') = L\big(RR' + \eta\flat_{\varphi'},R'^*\varphi + \varphi\cdot b^*,R\eta' + \eta\cdot b,\inner{\varphi}{\eta}_{\mathfrak{B}} + bb'\big)\in\mathfrak{L}
\end{align}
and so $\mathfrak{L}$ is closed under multiplication.

Now, we prove the rightmost inequality in Eq. (\ref{eq:Lbound}).  Note that
\begin{align}
L(R,\varphi,\eta,b) = \begin{pmatrix}
R & 0\\
0 & 0
\end{pmatrix} + \begin{pmatrix}
0 & \eta\\
0 & 0
\end{pmatrix} + \begin{pmatrix}
0 & 0\\
\flat_\varphi & 0
\end{pmatrix} + \begin{pmatrix}
0 & 0\\
0 & b
\end{pmatrix}
\end{align}
so that we can treat each summand separately.

For the first summand, we have
\begin{align}
\label{eq:bound1}
    \left \Vert \begin{pmatrix}
R & 0\nonumber\\
0 & 0
\end{pmatrix}\right \Vert  &= \sup_{\Vert (\xi,b')\Vert\leq 1}\left \Vert \begin{pmatrix}
R & 0\nonumber\\
0 & 0
\end{pmatrix} \begin{pmatrix} \xi\nonumber\\
b'\end{pmatrix}\right \Vert \nonumber\\
&= \sup_{\Vert (\xi,b')\Vert\leq 1} \left\Vert\begin{pmatrix}
R\xi\\
0\end{pmatrix}\right\Vert = \sup_{\Vert \xi\Vert_{\mathfrak{B}}\leq 1} \Vert
R\xi\Vert_{\mathfrak{B}} = \Vert R\Vert.
\end{align}

For the second summand, we have
\begin{align}
    \left\Vert\begin{pmatrix}
    0 & \eta\nonumber\\
    0 & 0
    \end{pmatrix}\right\Vert &= \sup_{\Vert (\xi,b')\Vert\leq 1} \left\Vert\begin{pmatrix}
    0 & \eta\nonumber\\
    0 & 0
    \end{pmatrix} \begin{pmatrix}
    \xi\\
    b'
    \end{pmatrix}\right\Vert\\
    &= \sup_{\Vert (\xi,b')\Vert\leq 1}\left\vert\begin{pmatrix}
    \eta\cdot b'\\
    0
    \end{pmatrix}\right\Vert = \sup_{\norm{b'}\leq 1} \norm{\eta\cdot b'}_{\mathfrak{B}} \leq \norm{\eta}_{\mathfrak{B}}.
\end{align}

For the third summand, we have from the Cauchy-Schwarz inequality \citep[][p. 3, Prop. 1.1]{La95} that
\begin{align}
    \left\Vert\begin{pmatrix}
    0 & 0\nonumber\\
    \flat_\varphi & 0
    \end{pmatrix}\right\Vert &= \sup_{\norm{(\xi,b')}\leq 1} \left\Vert \begin{pmatrix}
    0 & 0\nonumber\\
    \flat_\varphi & 0
    \end{pmatrix} \begin{pmatrix}
    \xi\nonumber\\
    b'\end{pmatrix}\right\Vert\nonumber\\
    &= \sup_{\norm{(\xi,b')}\leq 1} \left\Vert\begin{pmatrix}
    0\\
    \inner{\varphi}{\xi}_{\mathfrak{B}}
    \end{pmatrix}\right\Vert \leq \sup_{\norm{\xi}_{\mathfrak{B}}\leq 1} \norm{\inner{\varphi}{\xi}_{\mathfrak{B}}}\leq \norm{\varphi}_{\mathfrak{B}}.
\end{align}

For the fourth summand, we have
\begin{align}
    \left\Vert\begin{pmatrix}
    0 & 0\nonumber\\
    0 & b
    \end{pmatrix}\right\Vert &=\sup_{\norm{(\xi,b')}\leq 1}\left \Vert \begin{pmatrix}
    0 & 0\nonumber\\
    0 & b
    \end{pmatrix} \begin{pmatrix}
    \xi\nonumber\\
    b'
    \end{pmatrix}\right\Vert\nonumber\\
    &= \sup_{\norm{\xi,b'}\leq 1} \left\Vert\begin{pmatrix}
    0\\
    bb'
    \end{pmatrix}\right\Vert \leq \sup_{\norm{b'}\leq 1} \norm{bb'} \leq \norm{b}.
\end{align}
Hence, the rightmost inequality in Eq. (\ref{eq:Lbound}) is implied by the triangle inequality.

For the leftmost inequality in Eq. (\ref{eq:Lbound}), we consider each term.

First, since we have 
\begin{align}
    \begin{pmatrix}
    R & 0\\
    0 & 0
    \end{pmatrix} = L(R,\varphi,\eta,b) - \begin{pmatrix}
    0 & \eta\\
    0 & 0
    \end{pmatrix} - \begin{pmatrix}
    0 & 0\\
    \flat_{\varphi} & 0
    \end{pmatrix} - \begin{pmatrix}
    0 & 0\\
    0 & b
    \end{pmatrix},
\end{align}
it follows from Eq. (\ref{eq:bound1}) and the triangle inequality that $\norm{R}\leq\norm{L(R,\varphi,\eta,b)}$.

Second, note that it follows from the Cauchy-Schwarz inequality that
\begin{align}
    \left\langle \begin{pmatrix}
    \eta\nonumber\\
    0\end{pmatrix}
    , L(R,\varphi,\eta,b) \begin{pmatrix}
    0\nonumber\\
    \inner{\eta}{\eta}_{\mathfrak{B}}\end{pmatrix}
    \right\rangle_{\mathfrak{B}}^*&\left\langle \begin{pmatrix}
    \eta\nonumber\\
    0\end{pmatrix}
    , L(R,\varphi,\eta,b) \begin{pmatrix}
    0\nonumber\\
    \inner{\eta}{\eta}_{\mathfrak{B}}\end{pmatrix}
    \right\rangle_{\mathfrak{B}}\nonumber\\ 
    &\leq \left\Vert L(R,\varphi,\eta,b)\begin{pmatrix}
    0\nonumber\\
    \inner{\eta}{\eta}_{\mathfrak{B}}
    \end{pmatrix}\right\Vert_{\mathfrak{B}}^2 \left\langle\begin{pmatrix}
    \eta\nonumber\\
    0
    \end{pmatrix},\begin{pmatrix}
    \eta\nonumber\\
    0
    \end{pmatrix}\right\rangle_{\mathfrak{B}}\nonumber\\
    &\leq \norm{\eta}_{\mathfrak{B}}^4\norm{L(R,\varphi,\eta,b)}^2\inner{\eta}{\eta}_{\mathfrak{B}}.
\end{align}
Hence, we have
\begin{align}
    \norm{\eta}_{\mathfrak{B}}^4 = \norm{\inner{\eta}{\eta}_{\mathfrak{B}}}^2 = \left\Vert\left\langle \begin{pmatrix}
    \eta\\
    0
    \end{pmatrix}, L(R,\varphi,\eta,b)\begin{pmatrix}
    0\\
    \inner{\eta}{\eta}_{\mathfrak{B}}
    \end{pmatrix}\right\rangle_{\mathfrak{B}}\right\Vert\leq \norm{L(R,\varphi,\eta,b)}\norm{\eta}_{\mathfrak{B}}^3,
\end{align}
which implies $\norm{\eta}_{\mathfrak{B}}\leq\norm{L(R,\varphi,\eta,b)}$.

Third, we do a similar calculation to show
\begin{align}
    \left\langle \begin{pmatrix}
    0\nonumber\\
    \inner{\varphi}{\varphi}_{\mathfrak{B}}
    \end{pmatrix}, L(R,\varphi,\eta,b)\begin{pmatrix}
    \varphi\nonumber\\
    0
    \end{pmatrix}\right\rangle_{\mathfrak{B}}^*&\left\langle \begin{pmatrix}
    0\nonumber\\
    \inner{\varphi}{\varphi}_{\mathfrak{B}}
    \end{pmatrix}, L(R,\varphi,\eta,b)\begin{pmatrix}
    \varphi\nonumber\\
    0
    \end{pmatrix}\right\rangle_{\mathfrak{B}}\nonumber\\
    &\leq \left\Vert L(R,\varphi,\eta,b)\begin{pmatrix}
    \varphi\nonumber\\
    0
    \end{pmatrix}\right\Vert^2_{\mathfrak{B}} \left\langle\begin{pmatrix}
    0\nonumber\\
    \inner{\varphi}{\varphi}_{\mathfrak{B}}
    \end{pmatrix},\begin{pmatrix}
    0\nonumber\\
    \inner{\varphi}{\varphi}_{\mathfrak{B}}
    \end{pmatrix}\right\rangle_{\mathfrak{B}}\nonumber\\
    &\leq\norm{\varphi}_{\mathfrak{B}}^2\norm{L(R,\varphi,\eta,b)}^2\inner{\varphi}{\varphi}_{\mathfrak{B}}^2.
\end{align}
Hence, we have
\begin{align}
    \norm{\varphi}_{\mathfrak{B}}^4 = \norm{\inner{\varphi}{\varphi}_{\mathfrak{B}}}^2 = \left\Vert \left\langle \begin{pmatrix}
    0\\
    \inner{\varphi}{\varphi}_{\mathfrak{B}}
    \end{pmatrix}, L(R,\varphi,\eta,b) \begin{pmatrix}
    \varphi\\
    0
    \end{pmatrix}\right\rangle_{\mathfrak{B}}\right\Vert\leq \norm{L(R,\varphi,\eta,b)}\norm{\varphi}^3_{\mathfrak{B}},
\end{align}
which implies $\norm{\varphi}_{\mathfrak{B}}\leq\norm{L(R,\varphi,\eta,b)}$.

Fourth, since we have
\begin{align}
    \norm{b}^3 = \left\Vert\left\langle \begin{pmatrix}
    0\\
    b
    \end{pmatrix}, L(R,\varphi,\eta,b)\begin{pmatrix}
    0\\
    b
    \end{pmatrix}\right\rangle_{\mathfrak{B}}\right\Vert\leq \norm{L(R,\varphi,\eta,b)}\norm{b}^2,
\end{align}
it follows that $\norm{b}\leq\norm{L(R,\varphi,\eta,b)}$.

Thus, we have established both inequalities in Eq. (\ref{eq:Lbound}).  Now it follows that $\mathfrak{L}$ is norm closed as well because $\mathcal{K}(\mathcal{E})$, $\mathcal{E}$, $\mathfrak{B}$, and the space of operators of the form $\flat_\varphi$ are all closed.  Thus, $\mathfrak{L}$ is a C*-algebra.
\end{proof}

Next, we define a subalgebra of $\mathfrak{L}$.  We will consider the collection of compact operators $R$ on $\mathcal{E}$ whose image is contained in $\overline{\mathcal{E}\cdot K_{\alpha(\hbar)}}$ in the sense that for all $\psi\in\mathcal{E}$, $R\psi\in \overline{\mathcal{E}\cdot K_{\alpha(\hbar)}}$.  As shorthand for this property, we will say $R\in \mathcal{K}(\mathcal{E},\overline{\mathcal{E}\cdot K_{\alpha(\hbar)}})$ when $R\in\mathcal{K}(\mathcal{E})$ and $R(\mathcal{E})\subseteq \overline{\mathcal{E}\cdot K_{\alpha(\hbar)}}$.  Now we define
\begin{align}
    \mathfrak{D} = \big\{L(R,\varphi,\eta,b)\ |\ R\in \mathcal{K}(\mathcal{E},\overline{\mathcal{E}\cdot K_{\alpha(\hbar)}});\  b\in K_{\alpha(\hbar)};\ \varphi,\eta\in \overline{\mathcal{E}\cdot K_{\alpha(\hbar)}}\big\}.
\end{align}
\medskip

\begin{lemma}
$\mathfrak{D}$ is a closed two-sided ideal of $\mathfrak{L}$.
\end{lemma}

\begin{proof}
Consider $L = L(R,\varphi,\eta,b)\in \mathfrak{L}$ and $L'=L(R',\varphi',\eta',b')\in\mathfrak{D}$.  We have

\begin{align}
    LL' = \begin{pmatrix}
    R & \eta\\
    \flat_{\varphi} & b
    \end{pmatrix} \begin{pmatrix}
    R' & \eta'\\
    \flat_{\varphi'} & b'
    \end{pmatrix} = \begin{pmatrix}
    RR' + \eta\flat_{\varphi'} & R\eta' + \eta b'\\
    \flat_{\varphi}R' + b\flat_{\varphi'} & \inner{\varphi}{\eta'}_{\mathfrak{B}} + bb'
    \end{pmatrix}
\end{align}
and similarly
\begin{align}
    L'L = \begin{pmatrix}
    R'R + \eta'\flat_{\varphi} & R'\eta + \eta' b\\
    \flat_{\varphi'}R + b'\flat_{\varphi} & \inner{\varphi'}{\eta}_{\mathfrak{B}} + b'b
    \end{pmatrix}.
\end{align}
We consider each matrix element in turn and show that it belongs to the appropriate domain.

\begin{enumerate}
    \item $RR' + \eta \flat_{\varphi'}\in \mathcal{K}(\mathcal{E},\overline{\mathcal{E}\cdot K_{\alpha(\hbar)}})$ and $R'R + \eta'\flat_\varphi\in (\mathcal{K}(\mathcal{E},\overline{\mathcal{E}\cdot K_{\alpha(\hbar)}})$.

    Clearly, each of the operators in each of these sums is compact.  We must show that their images are contained in $\mathcal{K}(\mathcal{E},\overline{\mathcal{E}\cdot K_{\alpha(\hbar)}})$.
    
    Since $R'\in \mathcal{K}(\mathcal{E},\overline{\mathcal{E}\cdot K_{\alpha(\hbar)}})$, we have that for any $\psi\in\mathcal{E}$, $R'\psi = \lim_i\sum_i \xi_i b_i$ for some $b_i\in K_{\alpha(\hbar)}$ and $\xi_i\in \mathcal{E}$.  This implies immediately that $R'R\psi \in \overline{\mathcal{E}\cdot K_{\alpha(\hbar)}}$.  Moreover, since the right $\mathfrak{B}$-action commutes with all compact operators, we have that $RR'\psi = \lim_i\sum_i R(\xi_i b_i) = \lim_i\sum_i R\xi_i\cdot b_i\in \overline{\mathcal{E}\cdot K_{\alpha(\hbar)}}$.  Hence, we have that $RR',R'R\in \mathcal{K}(\mathcal{E},\overline{\mathcal{E}\cdot K_{\alpha(\hbar)}})$.

    Next, since $\varphi'\in \overline{\mathcal{E}\cdot K_{\alpha(\hbar)}}$, we have that $\varphi' = \lim_i\sum_i\xi_i b_i$ for some $b_i\in K_{\alpha(\hbar)}$ and $\xi_i\in\mathcal{E}$.  So for any $\psi\in\mathcal{E}$, we have 
    \begin{align}
    \eta\flat_{\varphi'}(\psi) = \eta\inner{\varphi'}{\psi}_{\mathfrak{B}}
    =\lim_i\sum_i \eta\inner{\xi_i b_i}{\psi}_{\mathfrak{B}}
    = \lim_i\sum_i \eta b_i^* \inner{\xi_i}{\psi}_{\mathfrak{B}},
    \end{align}
    and since $K_{\alpha(\hbar)}$ is an ideal, we have $\eta\flat_{\varphi'}(\psi)\in\overline{\mathcal{E}\cdot K_{\alpha(\hbar)}}$.

    Similarly, since $\eta'\in \overline{\mathcal{E}\cdot K_{\alpha(\hbar)}}$, we have that $\eta' = \lim_i\sum_i\xi_i b_i$ for some $b_i\in K_{\alpha(\hbar)}$ and $\xi_i\in\mathcal{E}$.  So for any $\psi\in\mathcal{E}$, we have
    \begin{align}
        \eta'\flat_{\varphi}(\psi) = \eta'\inner{\varphi}{\psi}_{\mathfrak{B}} = \lim_i\sum_i \xi_i b_i\inner{\varphi}{\psi}_{\mathfrak{B}},
    \end{align}
    and since $K_{\alpha(\hbar)}$ is an ideal, we have $\eta'\flat_{\varphi}(\psi)\in \overline{\mathcal{E}\cdot K_{\alpha(\hbar)}}$.  
    
Since $\mathcal{K}(\mathcal{E},\overline{\mathcal{E}\cdot K_{\alpha(\hbar)}})$ is a linear subspace of $\mathcal{L}(\mathcal{E})$, this establishes that the first matrix element of each of $LL'$ and $L'L$ belongs to $\mathcal{K}(\mathcal{E},\overline{\mathcal{E}\cdot K_{\alpha(\hbar)}})$.

\item $R\eta' + \eta b'\in \overline{\mathcal{E}\cdot K_{\alpha(\hbar)}}$ and $R'\eta + \eta' b\in \overline{\mathcal{E}\cdot K_{\alpha(\hbar)}}$.

Since $\eta'\in \overline{\mathcal{E}\cdot K_{\alpha(\hbar)}}$, we have $\eta' = \lim_i\sum_i \xi_i b_i$ for some $\xi_i\in \mathcal{E}$ and $b_i\in K_{\alpha(\hbar)}$.  Since the right $\mathfrak{B}$-action commutes with all compact operators, we have that $R\eta' = \lim_i\sum_i R(\xi_i b_i) = \lim_i\sum_i R\xi_i\cdot b_i$.
Thus, $R\eta'\in \overline{\mathcal{E}\cdot K_{\alpha(\hbar)}}$.  And since $\overline{\mathcal{E}\cdot K_{\alpha(\hbar)}}$ is a linear subspace, we have $R\eta' + \eta b'\in \overline{\mathcal{E}\cdot K_{\alpha(\hbar)}}$.

On the other hand, since $R'\in \mathcal{K}(\mathcal{E},\overline{\mathcal{E}\cdot K_{\alpha(\hbar)}})$, we have $R'\eta\in \overline{\mathcal{E}\cdot K_{\alpha(\hbar)}}$.  And likewise, since $K_{\alpha(\hbar)}$ is an ideal, we have $\eta'b\in \overline{\mathcal{E}\cdot K_{\alpha(\hbar)}}$, which implies further that $R'\eta + \eta'b\in \overline{\mathcal{E}\cdot K_{\alpha(\hbar)}}$.

\item  $\flat_\varphi R' + b\flat_{\varphi'} = \flat_\xi$ and $\flat_{\varphi'} R + b' \flat_{\varphi} = \flat_{\xi'}$ for some $\xi,\xi'\in \overline{\mathcal{E}\cdot K_{\alpha(\hbar)}}$.

Define $\xi = R'^*\varphi + \varphi'\cdot b^*$.  It follows that for any $\psi\in\mathcal{E}$,
\begin{align}
    \flat_{\xi}\psi &= \inner{R'^*\varphi}{\psi}_{\mathfrak{B}} + \inner{\varphi'\cdot b^*}{\psi}_{\mathfrak{B}}\nonumber\\
    &= \inner{\varphi}{R'\psi}_{\mathfrak{B}} + b\inner{\varphi'}{\psi}_{\mathfrak{B}}\nonumber\\
    &= \flat_{\varphi} R'\psi + b\flat_{\varphi'}\psi.
\end{align}
Moreover, since $R'\in \mathcal{K}(\mathcal{E},\overline{\mathcal{E}\cdot K_{\alpha(\hbar)}})$, it follows that $R'^*\in \mathcal{K}(\mathcal{E},\overline{\mathcal{E}\cdot K_{\alpha(\hbar)}})$, so $R'^*\varphi\in \overline{\mathcal{E}\cdot K_{\alpha(\hbar)}}$.  Similarly, since $K_{\alpha(\hbar)}$ is an ideal, we have $\varphi'\cdot b^*\in \overline{\mathcal{E}\cdot K_{\alpha(\hbar)}}$.  Hence, $\xi\in \overline{\mathcal{E}\cdot K_{\alpha(\hbar)}}$.

Similarly, define $\xi' = R^*\varphi' + \varphi\cdot b'^*$.  It follows that for any $\psi\in\mathcal{E}$,
\begin{align}
    \flat_{\xi'}\psi &= \inner{R^*\varphi'}{\psi}_{\mathfrak{B}} + \inner{\varphi\cdot b'^*}{\psi}_{\mathfrak{B}}\nonumber\\
    &= \inner{\varphi'}{R\psi}_{\mathfrak{B}} + b'\cdot \inner{\varphi}{\psi}_{\mathfrak{B}}\nonumber\\
    &= \flat_{\varphi'}R\psi + b'\flat_{\varphi}\psi.
\end{align}
Moreover, since $R^*$ is a compact operator and the right $\mathfrak{B}$-action commutes with compact operators, we have $R^*\varphi'\in \overline{\mathcal{E}\cdot K_{\alpha(\hbar)}}$.  Further, we have $\varphi\cdot b'^*\in \overline{\mathcal{E}\cdot K_{\alpha(\hbar)}}$.  Hence, $\xi' \in \overline{\mathcal{E}\cdot K_{\alpha(\hbar)}}$.

\item $\inner{\varphi}{\eta'}_{\mathfrak{B}} + bb'\in  K_{\alpha(\hbar)}$ and $\inner{\varphi'}{\eta}_{\mathfrak{B}} + b'b\in K_{\alpha(\hbar)}$

The fact that $K_{\alpha(\hbar)}$ is an ideal implies that both $\inner{\varphi}{\eta'}_{\mathfrak{B}} + bb'\in K_{\alpha(\hbar)}$ and $\inner{\varphi'}{\eta}_{\mathfrak{B}} + b'b\in K_{\alpha(\hbar)}$.
\end{enumerate}

\noindent The fact that $\mathfrak{D}$ is closed is implied by the bounds in Lemma \ref{lemma:link} along with the fact that each of the sets $\mathcal{K}(\mathcal{E},\overline{\mathcal{E}\cdot K_{\alpha(\hbar)}})$, $K_{\alpha(\hbar)}$, and $\overline{\mathcal{E}\cdot K_{\alpha(\hbar)}}$ is closed.
\end{proof}

Since $\mathfrak{D}$ is a closed two-sided ideal in $\mathfrak{L}$, the quotient C*-algebra $\mathfrak{L}/\mathfrak{D}$ is well-defined with the canonical quotient norm
\begin{align}
    \norm{L + \mathfrak{D}}_{\mathfrak{L}/\mathfrak{D}} = \inf_{D\in\mathfrak{D}}\norm{L + D}
\end{align}
for all $L\in\mathfrak{L}$.

\section{Tensor Products of Non-Unital Commutative C*-algebras}
\label{app:prod}

In this appendix we establish a technical lemma about tensor products of non-unital commutative C*-algebras.  It is well known that the tensor product of unital commutative C*-algebras corresponds to the Cartesian product of their state spaces \citep[][p. 849]{KaRi97}.  We extend this result to non-unital commutative C*-algebras.
 
So suppose, as above, that $\mathfrak{A}_0$ and $\mathfrak{B}_0$ are non-unital commutative C*-algebras and their unique C*-tensor product \citep[][Lemma 11.3.5, p. 854]{KaRi97} is denoted $\mathfrak{A}_0\otimes\mathfrak{B}_0$.  Gelfand duality implies that $\mathfrak{A}_0\cong C_0(M)$ and $\mathfrak{B}_0\cong C_0(N)$ for some locally compact spaces $M$ and $N$.  We will show that the pure state space $P$ of $\mathfrak{A}_0\otimes \mathfrak{B}_0$ is homeomorphic to $M\times N$.

To that end, let $\{1_\lambda^{\mathfrak{A}_0}\}_{\lambda\in\Lambda}$ and $\{1_{\lambda'}^{\mathfrak{B}_0}\}_{\lambda'\in\Lambda'}$ denote approximate units of $\mathfrak{A}_0$ and $\mathfrak{B}_0$, respectively.  Then define maps $pr_M: P\to M$ and $pr_N: P\to N$ by
\begin{align}
pr_M(p)(a) &= \lim_{\lambda'} p\big(a\otimes 1_{\lambda'}^{\mathfrak{B}_0}\big)\nonumber\\
pr_N(p)(b) &= \lim_{\lambda} p\big(1_{\lambda}^{\mathfrak{A}_0}\otimes b\big)
\end{align}
for all $p\in P$, $a\in\mathfrak{A}_0$, and $b\in\mathfrak{B}_0$.  Each limit converges \citep[][Thm. 2.3, p. 162]{La76}, and the limit is independent of the chosen approximate unit.  We need to show that when $p$ is a pure state, so are $pr_M(p)$ and $pr_N(p)$.\medskip

\begin{proposition}
If $p\in P$, then $pr_M(p)\in\mathcal{P}(\mathfrak{A}_0)$ and $pr_N(p)\in\mathcal{P}(\mathfrak{B}_0)$.
\end{proposition}

\begin{proof}
Suppose $p$ is a pure state.  First we confirm that $pr_M(p)$ and $pr_N(p)$ are both states.  They are clearly linear functionals. 
 To see that they are each positive, note that for $a\in\mathfrak{A}_0$ and $b\in\mathfrak{B}_0$, we have that $a^*a\otimes 1_{\lambda'}^{\mathfrak{B}_0}$ and $1_{\lambda}^{\mathfrak{A}_0}\otimes b^*b$ are both positive in $\mathfrak{A}_0\otimes\mathfrak{B}_0$ \citep[][Thm. 4.2.2(iv)]{KaRi97}.  Hence,
 \begin{align}
     pr_M(p)(a^*a) = \lim_{\lambda'} p\big(a^*a\otimes 1_{\lambda'}^{\mathfrak{B}_0}\big)\geq 0\nonumber\\
     pr_N(p)(b^*b) = \lim_{\lambda} p\big(1_{\lambda}^{\mathfrak{A}_0}\otimes b^*b\big)\geq 0.
\end{align}
To see that $pr_M(p)$ and $pr_N(p)$ are both normalized, note that $\{1_{\lambda}^{\mathfrak{A}_0}\otimes 1_{\lambda'}^{\mathfrak{B}_0}\}_{(\lambda,\lambda')\in\Lambda\times\Lambda'}$ is an approximate identity for $\mathfrak{A}_0\otimes \mathfrak{B}_0$ \citep[][Lemma 2.1, p. 161]{La76}.  Hence,
\begin{align}
    1 = \norm{p} = \lim_{(\lambda,\lambda')} p\big(1_{\lambda}^{\mathfrak{A}_0}\otimes 1_{\lambda'}^{\mathfrak{B}_0}\big)
\end{align}
and so it follows that
\begin{align}
    \norm{pr_M(p)} &= \lim_{\lambda'} pr_M(p)\big(1_{\lambda}^{\mathfrak{A}_0}\big) = \lim_\lambda\lim_{\lambda'} p\big(1_{\lambda}^{\mathfrak{A}_0}\otimes 1_{\lambda'}^{\mathfrak{B}_0}\big) = 1\nonumber\\
    \norm{pr_N(p)} &= \lim_\lambda pr_N(p)\big(1_{\lambda'}^{\mathfrak{B}_0}\big) = \lim_{\lambda'}\lim_{\lambda} p\big(1_{\lambda}^{\mathfrak{A}_0}\otimes 1_{\lambda'}^{\mathfrak{B}_0}\big) = 1.
\end{align}
Thus, $pr_M(p)$ and $pr_N(p)$ are both states.

We now need to show that each state is pure. Consider $a_1,a_2\in\mathfrak{A}_0$ and $b_1,b_2\in\mathfrak{B}_0$.  Note that $\{(1_{\lambda}^{\mathfrak{A}_0})^2\}_{\lambda\in\Lambda}$ and $\{(1_{\lambda'}^{\mathfrak{B}_0})^2\}_{\lambda'\in\Lambda'}$ are both also approximate units, so we have
\begin{align}
    pr_M(p)(a_1a_2) &= \lim_{\lambda'} p\big(a_1a_2\otimes (1_{\lambda'}^{\mathfrak{B}_0})^2\big)\nonumber\\
    &=\lim_{\lambda'} p\big((a_1\otimes 1_{\lambda'}^{\mathfrak{B}_0})\cdot(a_2\otimes 1_{\lambda'}^{\mathfrak{B}_0})\big)\nonumber\\
    &= \lim_{\lambda'} p\big(a_1\otimes 1_{\lambda'}^{\mathfrak{B}_0}\big)\cdot p\big(a_2\otimes 1_{\lambda'}^{\mathfrak{B}_0}\big)\nonumber\\
    &= pr_M(p)(a_1) \cdot pr_M(p)(a_2)
\end{align}
and
\begin{align}
    pr_N(p)(b_1b_2) &= \lim_{\lambda'} p\big( (1_{\lambda}^{\mathfrak{A}_0})^2\otimes b_1b_2\big)\nonumber\\
    &=\lim_{\lambda'} p\big((1_{\lambda}^{\mathfrak{A}_0}\otimes b_1)\cdot(1_{\lambda}^{\mathfrak{A}_0}\otimes b_2)\big)\nonumber\\
    &= \lim_{\lambda'} p\big(1_{\lambda}^{\mathfrak{A}_0}\otimes b_1\big)\cdot p\big(1_{\lambda}^{\mathfrak{A}_0}\otimes b_2\big)\nonumber\\
    &= pr_N(p)(b_1) \cdot pr_N(p)(b_2).
\end{align}
Hence, $pr_M(p)$ and $pr_N(p)$ are each multiplicative, and therefore, pure \citep[][Prop. 4.4.1, p. 269]{KaRi97}

\end{proof}

So we now have a map $(pr_M,pr_N): P\to M\times N$.  We proceed to define an inverse map $(pr_M,pr_N)^{-1}: M\times N\to P$ by the continuous linear extension of
\begin{align}
    (pr_M,pr_N)^{-1}(m,n)(a\otimes b) = m(a)n(b)
\end{align}
for all $a\in\mathfrak{A}_0$ and $b\in\mathfrak{B}_0$.  Again, we must show that when $m$ and $n$ are pure states on $\mathfrak{A}_0$ and $\mathfrak{B}_0$, then $(pr_M,pr_N)^{-1}(m,n)$ is a pure state on $\mathfrak{A}_0\otimes\mathfrak{B}_0$.\medskip

\begin{proposition}
    If $m\in\mathcal{P}(\mathfrak{A}_0)$ and $n\in\mathcal{P}(\mathfrak{B}_0)$, then $(pr_M,pr_N)^{-1}(m,n)\in P$.
\end{proposition}

\begin{proof}
    In what follows, we let $\dot{\mathfrak{A}}_0 = \mathfrak{A}_0\oplus\mathbb{C}$ and $\dot{\mathfrak{B}}_0 = \mathfrak{B}_0\oplus\mathbb{C}$ denote the unitzations of $\mathfrak{A}_0$ and $\mathfrak{B}_0$, respectively.  We know that $m$ has a unique extension to a pure state $\dot{m}$ on $\dot{\mathfrak{A}}_0$ and $n$ has a unique extension to a pure state $\dot{n}$ on $\dot{\mathfrak{B}}_0$ satisfying
    \begin{align}
        \dot{m}(a+\alpha I) &= m(a) + \alpha\nonumber\\
        \dot{n}(b + \alpha I) &= n(b) + \alpha
    \end{align}
    for $\alpha\in\mathbb{C}$, $a\in\mathfrak{A}_0$, and $b\in\mathfrak{B}_0$.  Hence, there is a unique state $\dot{m}\otimes\dot{n}$ on $\dot{\mathfrak{A}}_0\otimes\dot{\mathfrak{B}}_0$ such that
    \begin{align}
    (\dot{m}\otimes\dot{n})\big((a+\alpha I)\otimes(b+\alpha' I)\big) = \dot{m}(a + \alpha I)\cdot \dot{n}(b+\alpha' I)
    \end{align}
    for all $\alpha,\alpha'\in\mathbb{C}$, $a\in\mathfrak{A}_0$, and $b\in\mathfrak{B}_0$ \citep[][Prop. 11.1.1, p. 802]{KaRi97}.  Since $m$ and $n$ are pure states, so are $\dot{m}$ and $\dot{n}$, and it follows that $\dot{m}\otimes\dot{n}$ is a pure state as well \citep[][Prop. 11.3.2(ii), p. 848]{KaRi97}.

    Now consider the states at infinity $\iota_M\in\mathcal{P}(\dot{\mathfrak{A}}_0)$ and $\iota_N\in\mathcal{P}(\dot{\mathfrak{B}}_0)$ defined by
    \begin{align}
        \iota_M(a+\alpha I) &= \alpha\nonumber\\
        \iota_N(b+\alpha I) &= \alpha
    \end{align}
    for all $\alpha\in\mathbb{C}$, $a\in\mathfrak{A}_0$, and $b\in\mathfrak{B}_0$.  Similarly, define $\iota_{M\times N} = \iota_M\otimes \iota_N\in\mathcal{P}(\dot{\mathfrak{A}}_0\otimes\dot{\mathfrak{B}}_0)$.  By Thm. 2.3 of \citet[][p. 163]{La76}, every state on $\dot{\mathfrak{A}}_0\otimes\dot{\mathfrak{B}}_0$ is a convex combination
    \begin{align}
        \alpha_1\cdot f + \alpha_2\cdot (\dot{m}'\otimes \iota_N) + \alpha_3\cdot(\iota_M\otimes \dot{n}') + \alpha_4\cdot \iota_{M\times N}
    \end{align}
    for some states $\dot{m}'$ on $\mathfrak{A}_0$, $\dot{n}'$ on $\mathfrak{B}_0$, and $f$ on $\mathfrak{A}_0\otimes\mathfrak{B}_0$.  Since $\dot{m}\otimes\dot{n}$ is pure, it must be equal to exactly one component of such a convex sum.  However, states of the latter three forms annihilate all operators of the form $a\otimes b$ for $a\in\mathfrak{A}_0$ and $b\in\mathfrak{B}_0$, while $\dot{m}\otimes\dot{n}$ does not.  Hence, $\dot{m}\otimes\dot{n}$ is the unique extension of the state $(pr_M,pr_N)^{-1}(m,n)$ on $\mathfrak{A}_0\otimes\mathfrak{B}_0$, which therefore must be pure.
\end{proof}

\begin{proposition}
    The map $(pr_M,pr_N)^{-1}$ is an inverse to $(pr_M,pr_N)$ in the sense that for any $p\in P$, we have
    \begin{align}
        (pr_M,pr_N)^{-1}(pr_M(p),pr_N(p)) = p,
    \end{align}
    and for any $m\in M$ and $n\in N$, we have
    \begin{align}
        pr_M\circ(pr_M,pr_N)^{-1}(m,n) &= m\nonumber\\
        pr_N\circ(pr_M,pr_N)^{-1}(m,n) &= n.
    \end{align}
\end{proposition}

\begin{proof}
    For any $a\in\mathfrak{A}_0$ and $b\in\mathfrak{B}_0$, we have
    \begin{align}
        (pr_M,pr_N)^{-1}(pr_M(p),pr_N(p))\big(a\otimes b\big) &= \lim_{\lambda'} p\big(a\otimes 1_{\lambda'}^{\mathfrak{B}_0}\big) \cdot \lim_{\lambda} p(1_{\lambda}^{\mathfrak{A}_0}\otimes b\big)\nonumber\\
        &= \lim_{\lambda'}\lim_\lambda p\big((a\otimes b)\cdot (1_{\lambda}^{\mathfrak{A}_0}\cdot 1_{\lambda'}^{\mathfrak{B}_0})\big)\nonumber\\
        &= p\big(a\otimes b\big).
    \end{align}
    Similarly, we have
    \begin{align}
         pr_M\circ(pr_M,pr_N)^{-1}(m,n)(a) &= \lim_{\lambda'} m(a)\cdot n(1_{\lambda'}^{\mathfrak{B}_0}) = m(a)\nonumber\\
        pr_N\circ(pr_M,pr_N)^{-1}(m,n)(b) &= \lim_\lambda m(1_{\lambda}^{\mathfrak{A}_0})\cdot n(b) = n(b).
    \end{align}
\end{proof}

\begin{proposition}
    The maps $pr_M$, $pr_N$ and $(pr_M,pr_N)^{-1}$ are all continuous in the weak* topologies.
\end{proposition}

\begin{proof}
    Suppose that we have a net $\{p_\alpha\}$ in $P$ with $p_\alpha\to p$ in the weak* topology.  Then for every $a\in\mathfrak{A}_0$ and $b\in\mathfrak{B}_0$, we have $p_\alpha(a\otimes b)\to p(a\otimes b)$.  Hence,
    \begin{align}
        pr_M(p_\alpha)(a) = \lim_{\lambda'}p_\alpha(a\otimes 1_{\lambda'}^{\mathfrak{B}_0}) \to \lim_{\lambda'} p(a\otimes 1_{\lambda'}^{\mathfrak{B}_0}) = pr_M(p)(a)
    \end{align}
    and
    \begin{align}
        pr_N(p_\alpha)(b) = \lim_{\lambda} p_{\alpha}(1_{\lambda}^{\mathfrak{A}_0}\otimes b) \to \lim_{\lambda} p(1_{\lambda}^{\mathfrak{A}_0}\otimes b) = pr_N(p)(b).
    \end{align}
Similarly, suppose we have a net $\{m_\alpha\}$ in $M$ with $m_\alpha\to m$ and a net $\{n_{\alpha'}\}$ in $N$ with $n_{\alpha'}\to N$.  Then for every $a\in\mathfrak{A}_0$, we have $m_\alpha(a)\to m(a)$ and for every $b\in\mathfrak{B}_0$, we have $n_{\alpha'}(b)\to n(b)$.  Hence,
\begin{align}
    (pr_M,pr_N)^{-1}(m_\alpha,n_{\alpha'})(a\otimes b) &= m_\alpha(a)\cdot n_{\alpha'}(b)\nonumber\\
    &\to m(a)\cdot n(b) = (pr_M,pr_N)^{-1}(m,n)(a\otimes b).
\end{align}
\end{proof}
\noindent Hence, we have established that $P$ is homeomorphic to $M\times N$.

\bibliography{bibliography.bib}
\bibliographystyle{apalike}
\end{document}